\documentclass[twoside,11pt]{article}

\usepackage{blindtext}

\usepackage[preprint]{jmlr2e}

\usepackage{amsmath}
\usepackage{amssymb}
\usepackage{mathtools}
\usepackage{graphicx}
\usepackage{tabularx}
\usepackage{array}
\usepackage{subcaption}
\usepackage{hyperref}
\usepackage{algorithm}
\usepackage{algpseudocode}
\usepackage{enumitem}
\usepackage{natbib}
\usepackage{tabularray}

\newcolumntype{M}{>{\centering\arraybackslash}X}
 
\algnewcommand\algorithmicforeach{\textbf{for each}}
\algdef{S}[FOR]{ForEach}[1]{\algorithmicforeach\ #1\ \algorithmicdo}
\newtheorem{assumption}{Assumption}[theorem]

\usepackage{lastpage}
\jmlrheading{26}{2025}{1-\pageref{LastPage}}{3/25}{}{25-0471}{Josue N Rivera and Dengfeng Sun}

\ShortHeadings{Hamiltonian-Informed Optimal Neural Control and State Estimation}{Rivera and Sun}
\firstpageno{1}

\begin{document}

\title{Receding Hamiltonian-Informed Optimal Neural Control and State Estimation for Closed-Loop Dynamical Systems}

\author{\name Josue N. Rivera, Dengfeng Sun \email \{river264, dsun\}@purdue.edu \\
       \addr School of Aeronautics and Astronautics\\
       Purdue University\\
       West Lafayette, IN 47907, USA}

\editor{My editor}

\maketitle

\begin{abstract}
This paper formalizes Hamiltonian-Informed Optimal Neural (Hion) controllers, a novel class of neural network-based controllers for dynamical systems and explicit non-linear model-predictive control. Hion controllers estimate future states and develop an optimal control strategy using Pontryagin's Maximum Principle. The proposed framework, along with our Taylored Multi-Faceted Approach for Neural ODE and Optimal Control (T-mano) architecture, allows for custom transient behavior, predictive control, and closed-loop feedback, addressing limitations of existing methods. Comparative analyses with established model-predictive controllers revealed Hion controllers' superior optimality and tracking capabilities. Optimal control strategies are also demonstrated for both linear and non-linear dynamical systems.
\end{abstract}


\section{Introduction}

Optimal control problems often involve designing controllers for systems with complex, chaotic, and/or non-linear dynamics. These problems are crucial for unmanned aerial vehicles (UAVs) flight controllers, robotics, and nuclear power plants \citep{salzmann2023real, katayama2023model, naimi2022nonlinear}. Various methods have been developed to address these problems. Solutions include dynamic programming, bang-bang controllers, proportional-integral-derivative (PID) controllers, linear-quadratic regulators (LQR), reinforcement learning (RL), and many variants of model-predictive control (MPC). However, these methods often encounter challenges in delivering solutions that are both optimally effective and practical. Some methods react to deviations without considering the optimality of the control, while others can be expensive to operate in practice \citep{schwenzer2021review, bemporad2002model}. Neural network approaches also grapple with their own unique challenges to generate solutions that consider accurate system dynamics. The quality of the control is often contingent on the quality of the training data \citep{zheng2023physics}. Among the developed methods, MPCs are intriguing as they consider the effects of current control actions on future states. Nonetheless, many methods fail to address optimality conditions or computational efficiency when real-time optimization of the control is required \citep{bemporad2002model}.

To address these challenges, this paper introduces a new class of neural network-based controllers for dynamical systems: Hamiltonian-Informed Optimal Neural (Hion) controllers, along with a novel architecture, the Taylored Multi-Faceted Approach for Neural ODE and Optimal Control (T-mano). Hion controllers are a type of explicit non-linear MPC neural network-based models that map an observed and desired state to a continuous control strategy and expected future states. The objective is to optimize the parameters of the controller to provide state estimation and control that not only adheres to a system dynamics but also follows a given transient response profile. The model is intended to operate in a closed-loop system, where it can cope with delays in receiving state information by predicting the system's expected future behavior. Hion controllers offer a new alternative to RL methods and other MPC-based approaches for controlling dynamical systems.

\subsection{Background}
\index{Model-predictive control}

Model-predictive control (MPC) encompasses a set of algorithms that leverage future state estimation to formulate control strategies for dynamical systems \citep{schwenzer2021review}. These algorithms typically involve iteratively solving control optimization problems over a receding horizon within a closed-loop framework. Classical MPC, conventionally, consists of discretizing a system dynamics and forming an optimization problem that must be solved at each iteration \citep{kouvaritakis2016model}. The optimization requires determining a zero-order hold control solution at the current iteration and the next few ones (known as the control horizon) while estimating the new states to be observed. Zero-order hold control strategies may be solved through a multitude of techniques including quadratic and dynamic programming. However, several issues arise. First, optimization problems may be computationally expensive, specially when involving non-linear future state estimation and/or long horizon. Secondly, the sampling size must be significant for the optimization to conclude which may not provide the quick reaction needed for a volatile systems. Additionally, the derivation of zero-order hold control strategies may not provide the most optimal solution to the problem. Successive linearization based MPCs (SLMPCs) are a set of approaches that linearize a plant's dynamics at each iteration, facilitating quick state estimation for control optimization \citep{zhakatayev2017successive, kuhne2004model}. While the methods reduce the computational burden, they often compromises the accuracy of the estimation. Non-linear MPCs (NMPCs) are a later attempt to account for the inherent non-linearity of the system dynamics by employing surrogate models to reduce the computational cost of the state estimation. Explicit MPCs, which overlap with both linear and non-linear approaches, precompute control optimizations and reuse the solutions during real-time operation \citep{schwenzer2021review, bemporad2002model}. The controller proposed in this research falls within the category of explicit NMPC, combining the benefits of non-linear state estimation with the computational efficiency brought by using precomputed solutions.

Artificial neural networks with MPC (ANN-MPCs) represent a subset of MPC models that involve utilizing neural networks as prediction models in the control loop \citep{wang2021model, pang2023virtual, hewing2020learning, cavagnari1999neural}. Typically, the models are trained to predict the expected future state of the system after a specified time interval, relying on trajectory data collected from simulations or test environments. The primary advantage lies in the models' ability to significantly reduce the computational cost associated with classical prediction models by leveraging the approximations provided by neural networks. As a result, less time is needed before the next state of the system can be sampled and an action can be taken. Recurrent neural networks (RNNs) are employed in a subset of these algorithms \citep{jordanou2021echo, ren2022tutorial}. They involve feeding previously sampled states back into the model during inference and passing knowledge onto subsequent iterations to enhance predictive capabilities. Although a significant number of strategies focus on replacing the predictive component, a subset of ANN-MPCs attempt to replace the controller entirely by training them to be surrogate using collected trajectories and the corresponding control observed from a larger computationally expensive MPC model \citep{rivera2024fast, hertneck2018learning, aakesson2006neural}. However, the use of neural networks as surrogate models can impact the optimality, accuracy of the dynamics, and out-of-distribution performance.

Physics-informed neural networks with MPC (PINN-MPCs) aim to further enhance the capabilities of the predictive models in the control loop. At their core, they integrate information about the dynamics, resulting in refined future state estimation and the information passed to control optimization \citep{antonelo2024physics, faria2024data, arnold2021state, zheng2023physics, nicodemus2022physics}. Physics-Informed Neural Nets for Control (PINC), introduced by \cite{antonelo2024physics}, was one of the first methods proposed using PINNs as the prediction model in an MPC framework. The approach involves training a PINN to predict a continuous estimation of future states over a fixed horizon given a constant control signal. The method provides more reliable predictions that can guide control optimization. Additionally, due to the continuous state prediction, distinct control optimization strategies that rely on different sampling rates can be implemented or tested with a single PINN model. However, there are some limitations attached to these approaches. Although a continuous set of future states of the system is estimated, PINN-MPCs often rely on a single predicted state to guide the control, due to the MPC model's reliance on zero-order hold control strategies.


Few works exist that consider neural networks as physics-informed controller in dynamical systems. Fewer works exist that allow for the adjustment of the system's transient characteristic. Ours and newer approaches reinvents the idea of a controller when a neural network model is involved. \cite{schiassi2022bellman} illustrate the feasibility of training neural controllers using Bellman optimality principle, and how they can be extended using X-TFC for different initial and final conditions \citep{schiassi2021physics}. \citet{d2021pontryagin, barry2022physics, chi2024nodec, kamtue2024pontryagin} are recent attempts to establish neural network models that encourages PMP optimality. The works demonstrated that PMP can be used to train a neural network model to predict an optimal trajectory with desired transient properties. However, several limitations exist. One such limitation is that each model was only demonstrated to solve TVBNP for a single predefined initial and final state. This makes them impractical for closed-loop systems control where they will vary. Any new boundary condition would require fine-tuning the model (e.g., using X-TFC and/or retraining it). Building on top of these projects, our model provides explicit NMPC neural network-based control and state estimation that generalizes for variable inputs in a closed-loop. We also theoretically define a set of dynamical systems for which the proposed model is most effective. Hion controller removes the classical control optimization step conducted by classical MPC models, ANN-MPC, and PINN-MPC, and uses a single neural network as both the prediction and control model.

Figure \ref{fig:closed-controller} illustrates a comparison between the conceptual behavior of different model-predictive controllers.

\begin{figure}[ht]
  \centering
      \subcaptionbox
      {Classical MPC.}
      {\bfseries \includegraphics[width=0.24\textwidth]{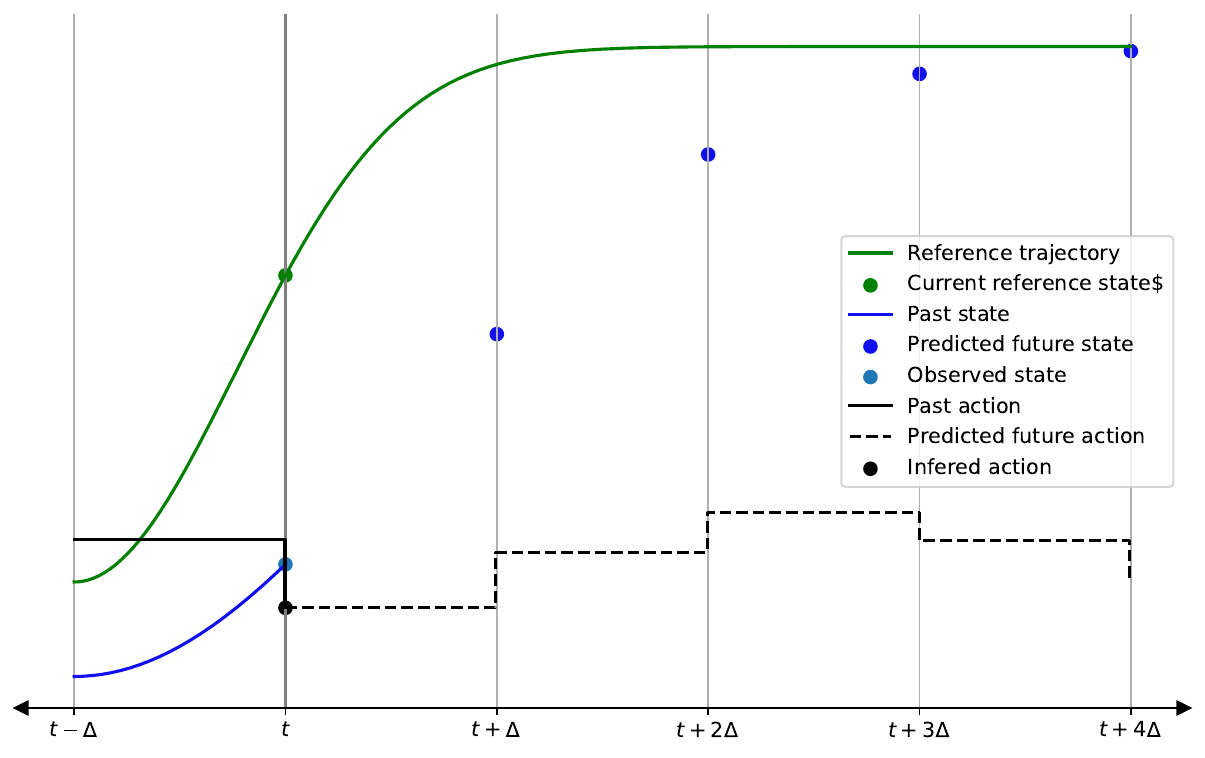}}%
    \subcaptionbox
      {ANN-MPC.}
      {\bfseries \includegraphics[width=0.24\textwidth]{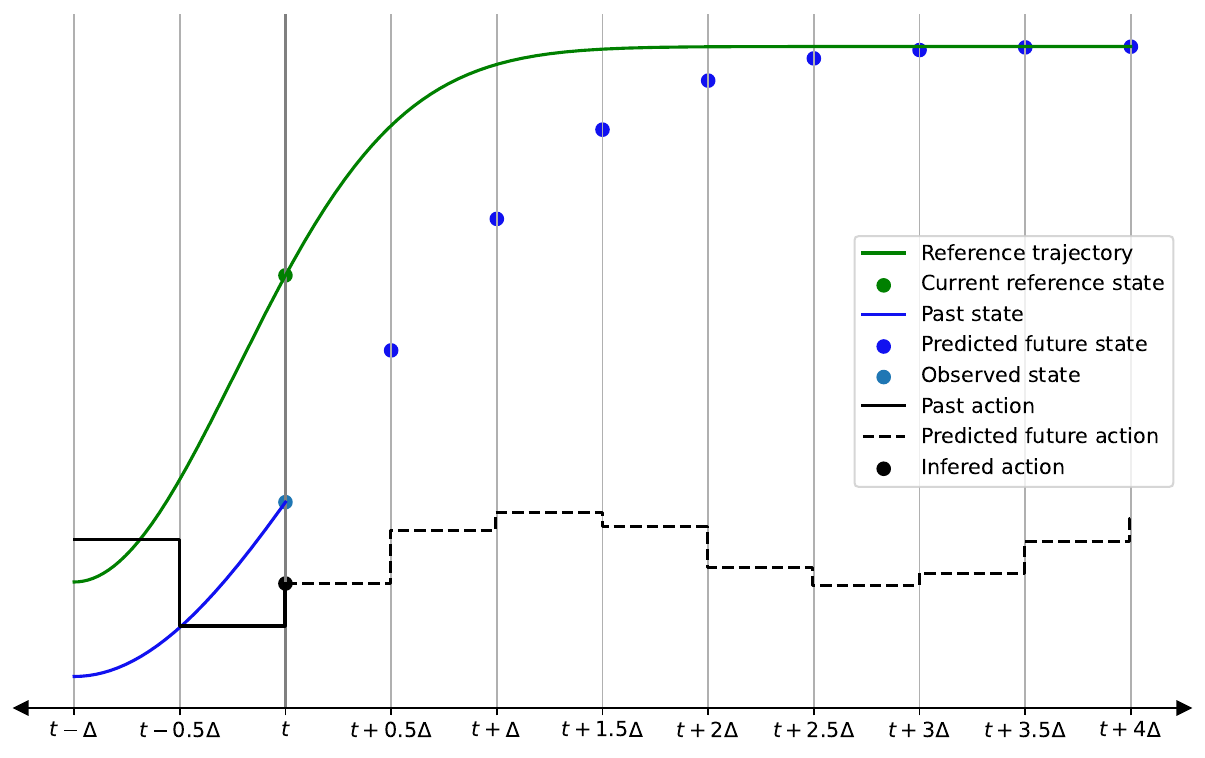}}
      \subcaptionbox
      {PINN-MPC.}
      {\bfseries \includegraphics[width=0.24\textwidth]{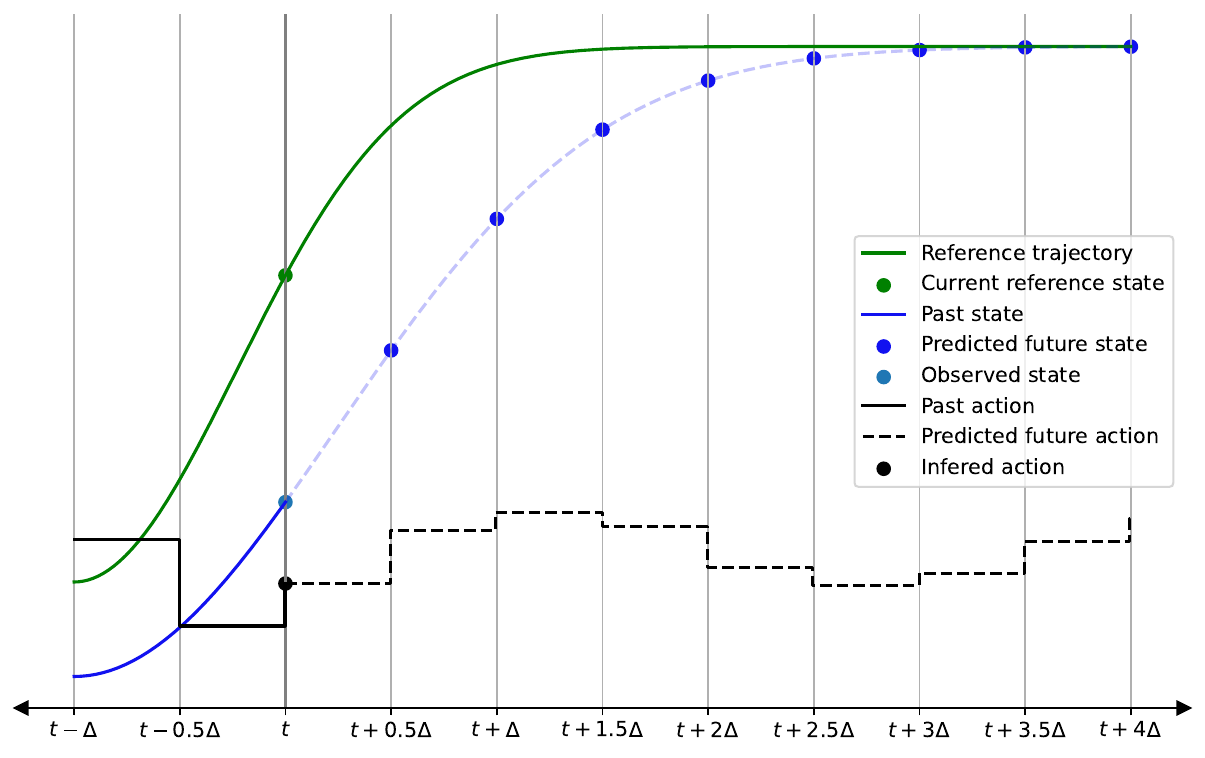}}%
    \subcaptionbox
      {Hion (Our).}
      {\bfseries \includegraphics[width=0.24\textwidth]{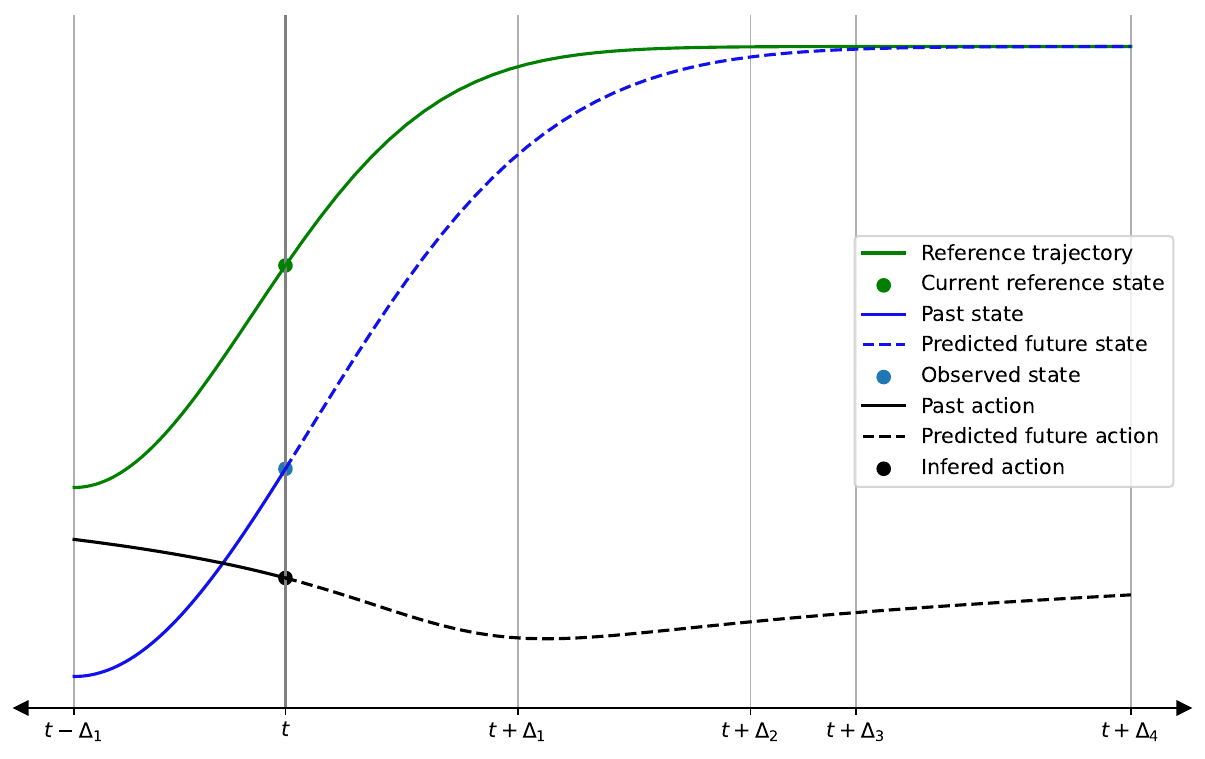}}%
      
    \caption{Conceptual behaviors of closed-loop model-predictive controllers.}
    \label{fig:closed-controller}
\end{figure}

\subsection{Contributions}

Our research formalizes a class of neural network controllers, termed Hion controllers, designed to generate control strategies within closed-loop dynamical systems. These controllers function as model-predictive controllers, enabling predictive decision-making. Our contributions include:

\begin{enumerate}
  
  \item Establishing a theoretical framework for Hion controllers.
  \item Proposing a novel neural network architecture specifically tailored for state estimation and control of dynamical systems.
  \item Developing algorithms for training Hion controllers that are aligned with Pontryagin's Maximum/Minimum Principle and encourage optimal control.
\end{enumerate}

\subsection{Outline}

Following the introduction, the first section formally defines the problem of interest we aim to address. We then present our proposed methodology and the underlying theoretical principles that support it. The experimental results section demonstrates the effectiveness and capabilities of our architecture through its application to various dynamical systems. Finally, we conclude by highlighting the advantages and limitations of the method along with potential future directions of our work.
\section{Problem Statement}

Consider a general dynamical system (or environment) with ordinary differential equations (ODEs) and a corresponding state-space representation describing it,

\begin{equation}
    \label{eq:system-dynamics}
\begin{array}{c}
    \mathcal{F}(t, \bar{x}(t), u(t)) = 0\\
    \dot{x}(t) = f(t, x(t), u(t))
\end{array}
\end{equation}
where $t$ represents time, $x(t)$ the state vector of the environment, $\bar{x}(t)$ the state vector in addition to some higher-order derivatives w.r.t. time needed to describe the ODEs, $u(t)$ the control action vector, $\mathcal{F}$ the ODEs describing the dynamics, and $f$ the dynamics function. Note $\dot{x}(t) = \frac{d}{d t} x(t)$.

\begin{definition} [Hion Controller]
Let \underline{H}amiltonian-\underline{I}nformed \underline{O}ptimal \underline{N}eural (Hion) controllers be a class of neural networks models $h: \mathbb{R}\times\mathbb{R}^n\times\mathbb{R}^n \rightarrow \mathbb{R}^k\times\mathbb{R}^m\times\mathbb{R}^n$ that maps the elapsed time $\hat{t}$ since a state was last observed, the last observed state $x_o \coloneqq x(t-\hat{t})$, and a reference state $x_r(t)$ to an inferred state $x_h(t)$, control $u_h(t)$, and co-state (i.e., Lagrange multiplier) $\lambda_h(t)$. As a neural network, $h$ contains a set of learnable parameters $\mu$.
\end{definition}

Given a transient cost to influence the behavior of the system and defined by the Lagrangian function $L$
\begin{equation}
    \label{eq:control-cost}
    J = \int_{0}^{\infty} L(t, x(t), x_r(t), u(t)) \,dt
\end{equation},
the \textbf{problem of interest} is to optimize a Hion controller's parameters $\mu$ to provide an optimal control strategy $u_h$ that reduces the least-square-error (LSE) between expected future states of the system $x_h$ and the current reference state $x_r$ while adhering to the dynamics of the system \eqref{eq:system-dynamics} and, secondarily, minimizing the transient cost \eqref{eq:control-cost}. The LSE condition w.r.t. the reference state may be omitted if a reference is not required by the problem at hand.
\section{Methodology}

Our proposed optimization for a Hion controller consists of three majors components: a distribution for observed and reference states, the Hion controller, and a set of criterion based on Pontryagin's Minimum/Maximum Principle (PMP). The state distribution defines random variable vectors from which inputs to the model can be sampled that aid in the converge of the controller. The controller defines the structure of the model and provides inferred states, controls, and their corresponding co-states. Lastly, the criterion evaluates the controller and the parameters $\mu$ are updated, accordingly. An overview of the methodology is presented in Figure \ref{gr:hion-flowchart}.

\begin{figure}[H]
  \centering
  \includegraphics[width=\textwidth]{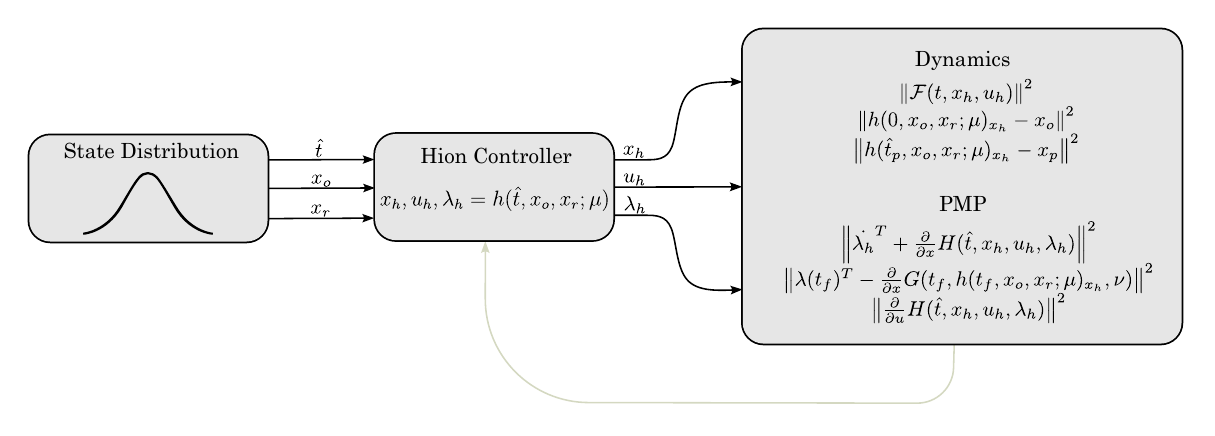}
  \caption{Hion Controller Training Flowchart}
  \label{gr:hion-flowchart}
\end{figure}

\subsection{State Distribution}

Following the convection set forth by PINNs \citep{hao2022physics}, the inputs to the Hion model are sampled from random variables. The first input, the elapsed time $\hat{t}$, is sampled from an uniform distribution
\begin{equation}
    \hat{T} \sim \mathcal{U}(0, t_f)
\end{equation}
where $t_f$ is a terminal time. In a two point boundary value problem\index{Two point boundary value problem (TPBVP)} (TPBVP), $t_f$ indicates the final time of the problem. In a closed-loop system, it represents the size of the time window for the system to be driven to a reference state. For the other inputs to the model, we must first introduce some definitions and theorems.

\begin{definition} [Dynamics invariant\index{Dynamics invariant} to a state transformation]
\label{de:dynamic-invariant}
Let $f$ be a dynamics function and $T$ be a transformation function. The dynamics $f$ is said to be invariant to a transformation $T$ on the state $x$ if and only if for all $x$, for all other distinct states $x'$, and for all possible control actions $u$, the following equality holds:
\[
f(x, x', u) = f(T(x), x', u)
\].
This means that the dynamics do not change when the state $x$ is transformed by $T$.
\end{definition}

Given Definition \ref{de:dynamic-invariant}, let us introduce the following theorems.

\begin{theorem}
\label{th:invariant-equivalent}
\index{Dynamics invariant}
The following statements are equivalent:

(A) The dynamics are invariant under all transformation $T$ applied to the state $x$.

(B) The dynamics function $f$ explicitly does not depend on the state $x$.

(C) Without loss of generality, a state $x$ in a dynamic function $f$ can always be assumed to be fixed at a particular value.
\end{theorem}

\begin{proof}
(A) $\Rightarrow$ (B): If the dynamics are invariant under all transformations $T$ applied to the state $x$, then by definition, for any transformation $T$ on the state $x$, the dynamics function $f(x, x', u)$ satisfies $f(x, x', u) = f(T(x), x', u)$ for all $x$, for all other states $x'$, and for all possible control actions $u$. This implies that the dynamics equation $f$ explicitly does not depend on the state $x$, because the dynamics remain the same even when $x$ is transformed by any $T$. Hence, $f(x, x', u) = F(x', u)$ and $f(T(x), x', u) = F(x', u)$ where $F(x', u)$ does not depend on $x$.

(B) $\Rightarrow$ (C): If the dynamics equation $f$ describing the system does not contain the state $x$, it means that the dynamics are the same for any state $x$. Therefore, without loss of generality, we can always choose our state $x$ so that it is always fixed a particular value.

(C) $\Rightarrow$ (A): If, without loss of generality, the state $x$ can always be assumed to be fixed at a particular value in a dynamics function, then when we can apply any transformation $T$ on the state $x$, and the assumption would still hold. Thus, the dynamics remain the same $f(x, x', u) = f(T(x), x', u)$. Hence, by definition, the dynamics are invariant to a transformed state $x$.

Therefore, we have shown that the statements (A), (B), and (C) are equivalent. This completes the proof. 
\end{proof}

\begin{theorem}
\label{th:translation-solution}
\index{Dynamics invariant}
Let $g(x_o, u)$ be the solution to $x(t)$ from the last observed state $x_o$ of the system and given the control actions $u$ during the elapse time. If the dynamics are invariant to all transformation $T$ on the state $\bar{x}$, then 

$$g(x_o, u)=g\left (\begin{bmatrix} \bar{x}_o\\ x'_o \end{bmatrix}, u\right ) = g\left(\begin{bmatrix} \bar{x}_o-\Delta_x\\ x'_o \end{bmatrix}, u\right)+\begin{bmatrix} \Delta_x\\ 0 \end{bmatrix}$$ 
where $\Delta_x$ is some value, and $x'$ all other distinct states.
\end{theorem}

\begin{proof}
    By Taylor expansion, the solution $g(x_o, u)$ can be expressed as,
    \begin{equation*}
    \begin{split}
    g(x_o, u) &= g\left (\begin{bmatrix} \bar{x}_o\\ x'_o \end{bmatrix}, u \right ) = \begin{bmatrix} \bar{x}_o\\ x'_o \end{bmatrix} + f(\bar{x}_o, x'_o, u)\, \hat{t} + \frac{\dot{f}(\bar{x}_o, x'_o, u)}{2!} \, \hat{t}^{\,2} + \sum_{n=3}^{\infty} \frac{f^{(n-1)}(\bar{x}_o, x'_o, u)}{n!} \, \hat{t}^{\,n}\\
    &= \begin{bmatrix} \bar{x}_o-\Delta_x\\ x'_o \end{bmatrix} + f(\bar{x}_o, x'_o, u)\, \hat{t} + \frac{\dot{f}(\bar{x}_o, x'_o, u)}{2!} \, \hat{t}^{\,2} + \sum_{n=3}^{\infty} \frac{f^{(n-1)}(\bar{x}_o, x'_o, u)}{n!} \, \hat{t}^{\,n} + \begin{bmatrix} \Delta_x\\ 0 \end{bmatrix}
    \end{split}
    \end{equation*}

    Given the dynamics are invariant to all transformation $T$ on the state $\bar{x}$, by Definition \ref{de:dynamic-invariant} and Theorem \ref{th:invariant-equivalent}, $f(\bar{x}_o, x'_o, u) = f(\bar{x}_o-\Delta_x, x'_o, u)$ and $f^{(i)}(\bar{x}_o, x'_o, u) = f^{(i)}(\bar{x}_o-\Delta_x, x'_o, u)$ for any $i^{th}$ derivative w.r.t. time and any arbitrary value $\Delta_x$. Hence,
    \begin{equation*}
    \begin{split}
    g(x_o, u) &= \begin{bmatrix} \Delta^o_x\\ x'_o \end{bmatrix} + f(\bar{x}_o, x'_o, u)\, \hat{t} + \frac{\dot{f}(\bar{x}_o, x'_o, u)}{2!} \, \hat{t}^{\,2} + \sum_{n=3}^{\infty} \frac{f^{(n-1)}(\bar{x}_o, x'_o, u)}{n!} \, \hat{t}^{\,n} + \begin{bmatrix} \Delta_x\\ 0 \end{bmatrix}\\
    &= \begin{bmatrix} \Delta^o_x\\ x'_o \end{bmatrix} + f(\Delta^o_x, x'_o, u)\, \hat{t} + \frac{\dot{f}(\Delta^o_x, x'_o, u)}{2!} \, \hat{t}^{\,2} + \sum_{n=3}^{\infty} \frac{f^{(n-1)}(\Delta^o_x, x'_o, u)}{n!} \, \hat{t}^{\,n} + \begin{bmatrix} \Delta_x\\ 0 \end{bmatrix}\\
    &= g\left (\begin{bmatrix} \Delta^o_x\\ x'_o \end{bmatrix}, u \right ) + \begin{bmatrix} \Delta_x\\ 0 \end{bmatrix} \\
    &= g\left (\begin{bmatrix} \bar{x}_o-\Delta_x\\ x'_o \end{bmatrix}, u \right ) + \begin{bmatrix} \Delta_x\\ 0 \end{bmatrix}
    \end{split}
    \end{equation*}
    where $\Delta^o_x := \bar{x}_o-\Delta_x$.
    
    Therefore, we have shown that,
    $$g\left (\begin{bmatrix} \bar{x}_o\\ x'_o \end{bmatrix}, u\right ) = g\left(\begin{bmatrix} \bar{x}_o-\Delta_x\\ x'_o \end{bmatrix}, u\right)+\begin{bmatrix} \Delta_x\\ 0 \end{bmatrix}$$
    for any arbitrary value $\Delta_x$. This completes the proof.
\end{proof}

With the theorems at hand, when training a Hion controller, we assume that the observed state's $x_o$ input must adhere to Assumption \ref{prop:state-dist-condition}.

\begin{assumption} [Observed State Condition]
    \label{prop:state-dist-condition}
  If a Hion controller converges to a set of parameter $\mu$ via parameters optimization while adhering to the dynamics $f$, then each observed state of the system must have either: (1) been sampled from a random variable with a proper distribution, or (2) the dynamics of the system are invariant to all transformation $T$ on the state.
\end{assumption}

The intuition behind Assumption \ref{prop:state-dist-condition} is that if an observed state corresponding random variable does not have a proper distribution, then training the model via sampling from the distribution may not converge. However, if the dynamics are invariant to all transformation on a given state, then during training, the observed state can be assigned some fixed particular value by Theorems \ref{th:invariant-equivalent} and \ref{th:translation-solution} assuming the dynamics hold for the inferred states $x_h$ and controls $u_h$. For inference, the invariant state can be fixed to the particular value, and the solution can be translated via Theorem \ref{th:translation-solution}. An example of states in a system that would satisfy the dynamics invariant property would be position in an unconstrained quadcopter \citep{abougarair2024dynamics}. States in an unconstrained resource allocation optimal control problem (e.g., queen-worker insect problem) would not satisfy the condition or have a proper distribution \citep{winkel2013discover, oster1978caste}.

The reference state's $x_r$ sampled for training are chosen from a random variable $\mathbf{X_r}$ that depends on the observed state training input's random variable $\mathbf{X_o}$. Unless otherwise considered, a reference state is always assumed to be feasible or nearly feasible to allow for tracking in a short time span. As such, if the state is continuous and unconstrained, the reference state input random variable can be assumed to be
$$
    \quad \mathbf{X_r} = x^r_o + \mathbf{E}, \quad \mathbf{E} \sim \mathcal{N}(0, \sigma^2)
$$
where $x^r_o$ is the subset of the observed states relevant to the reference, and $\sigma^2$ defines the variance of the reference states with respect to the observed states.

\subsection{Controller Architecture}

As presented in the problem statement, a Hion controller $h(\cdot)$ is described by
\begin{equation}
x_h, u_h, \lambda_h = h(\hat{t}, x_o, x_r; \mu) 
\end{equation}
where elapsed time $\hat{t}$, last observed state $x_o$, and reference state $x_r$ describe the inputs to the model, sampled from a distribution. On the other hand, the expected state $x_h$, corresponding control $u_h$, and co-state (i.e., Lagrange multiplier) $\lambda_h$ describe the generative output. $\mu$ is the set of parameters dictating the behavior of the model. 

As it is desired that the neural network adhere to the dynamics and akin to PINN models, primitive states (i.e., zero-order derivative w.r.t. time)\index{Primitive states} $x^h_{[0]}$ and controls $u^h_{[0]}$\index{Primitive controls} are inferred. Any higher-order derivative is obtained by differentiation.
\begin{equation}
\label{eq:hion-state-diff}
    x^h_{[i+1]} = \dot{x}^h_{[i]} = \left(x^h_{[0]}\right)^{(i)} = \frac{\partial}{\partial \hat{t}} \left[ x^h_{[i]} \right] \quad\quad
u^h_{[i+1]} = \dot{u}^h_{[i]} = \left(u^h_{[0]}\right)^{(i)} = \frac{\partial}{\partial \hat{t}} \left[ u^h_{[i]} \right]
\end{equation}

The output of the model, $x_h$ and $u_h$, denote the vector of the states $x^h_{[i]}$ and controls $u^h_{[i]}$ relevant to the system, respectively. These may, however, not follow accurate dynamics or be optimal when modeled by a conventional neural network model. 

\subsubsection{Taylored Multi-Faceted Approach for Neural ODE and Optimal Control}

The \textbf{problem of interest} present unique opportunity to build a neural network architecture tailored for the task. Taylored Multi-Faceted Approach for Neural ODE and Optimal Control (T-mano) is a novel architecture for optimal control and state estimation of dynamical systems based on Taylor expansion. It ensures the initial conditions and accuracy of some of the systems dynamics are maintained in inference regardless of the parameters optimization. The architecture is presented in Figure \ref{gr:t-mano} and consists of four main stages.

\begin{figure}[H]
  \centering
  \includegraphics[width=\textwidth]{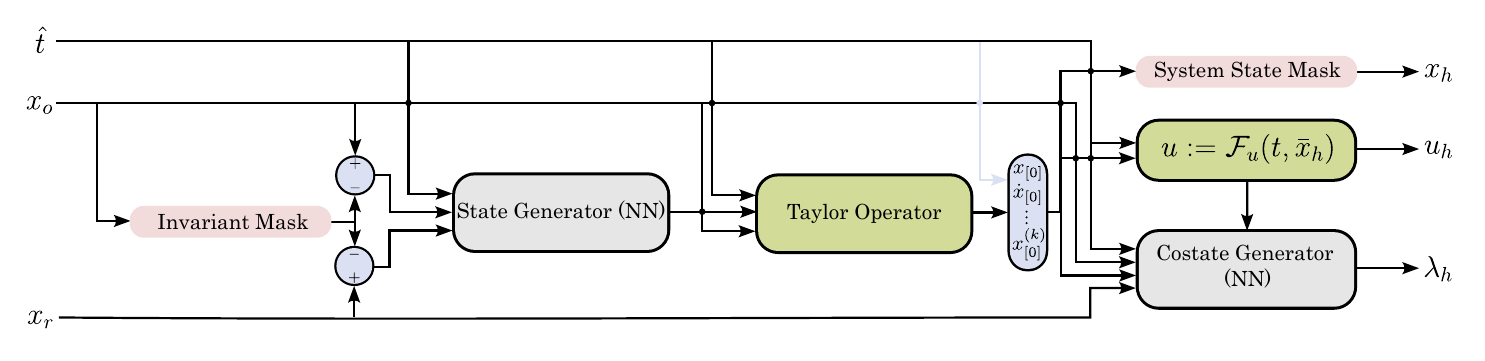}
  \caption{T-Mano Architecture for Control. It consists of four main stages that ensure that initial conditions and partial dynamics accuracy are maintained.}
  \label{gr:t-mano}
\end{figure}

Initially, the invariant mask filters observed states that are considered invariant to all transformation under Definition \ref{de:dynamic-invariant} and fixes them to zero before the state generator. The results does not affect the inference of the model as a solution to the system was proven to generalize to other invariant states in Theorem \ref{th:translation-solution} when fixed to a value. However, fixing these invariant states to a value (e.g., zero in this case) was empirically found to improve convergence and generalization to out-of-distribution invariant states.

In the first main stage, the state generator (as a neural network) provides a higher-order terms state function $\hat{x}_h(\hat{t})$ to the Taylor operator before any future states are inferred. It can be modeled by a multi-layer perceptron (MLP),
\begin{equation}
    \label{eq:hion-state-def}
    \begin{split}
        \hat{x}_h &= h_{x}(\hat{t}, \bar{x}_o, \bar{x}_r; \mu_x) = (h_k \circ \sigma \circ ... \circ h_2 \circ \sigma \circ h_1) (\hat{t}, \bar{x}_o, \bar{x}_r) \\
        h_j(z) &=  W_j z + b_j \\
    \end{split}
\end{equation}
where $\mu_x = \{W_1, b_1, ..., W_k, b_k\}$ is the set of all the learnable parameters of the model, and $\sigma(\cdot)$ is the non-linear activation function. For the experiments explored, $\sigma(\cdot) = \text{SiLU}(\cdot)$.

The Taylor operator $\mathcal{T}(\cdot)$ is a mapping of the observed state $x_o$ and a higher-order terms state function $\hat{x}(\hat{t})$ to a state estimation (or trajectory) that ensure that the initial condition of the system $x_o$ is respected. It is defined as:
\begin{equation}
    \mathcal{T}(\hat{t}, x_o, \hat{x}(\hat{t}))  = \sum_{n=0}^{k-1} \left[ \frac{x^o_{[n]}}{n!} \, \hat{t}^{\,n}\right ] + \hat{x}(\hat{t}) \, \hat{t}^{\,k}
\end{equation}
where $x^o_{[i]}$ represents the $i^{th}$-order derivatives w.r.t. time of the primitive state $x^o_{[0]}$ in the observed states $x_o$, $k-1$ is the highest known order, and $\hat{x}(\hat{t})$ is a higher-order terms state function. In the architecture, the higher-order terms state function is provided by the state generator $\hat{x}_h(\hat{t})$ and the operator is applied to each primitive state.

The theoretical basis for the Taylor operator comes from Taylor expansion that suggests that a solution $x_{[0]}(t)$ to the system can be expressed as:
\begin{equation}
    \label{eq:mlp-definition-state}
    \begin{split}
    x_{[0]}(t) &= x^o_{[0]} + x^o_{[1]}\, \hat{t} + \frac{x^o_{[2]}}{2!} \, \hat{t}^{\,2} + \sum_{n=3}^{\infty} \frac{x^o_{[n]}}{n!} \, \hat{t}^{\,n}
    \end{split}
\end{equation}.

However, given that observed state only contains $k-1$ known orders of state differentiability, the rest of the higher-order term states must be found.  Fortunately, also by Taylor expansion, they can be decomposed to a single unknown function $\hat{x}$:
\begin{equation}
    \begin{split}
    x_{[0]}(t) &= x^o_{[0]} + x^o_{[1]}\, \hat{t} + \frac{x^o_{[2]}}{2!} \, \hat{t}^{\,2} + \sum_{n=3}^{\infty} \frac{x^o_{[n]}}{n!} \, \hat{t}^{\,n} = \sum_{n=0}^{\infty} \left[ \frac{x^o_{[n]}}{n!} \, \hat{t}^{\,n}\right]\\
    &= \sum_{n=0}^{k-1} \left[ \frac{x^o_{[n]}}{n!} \, \hat{t}^{\,n}\right ] + \sum_{n=k}^{\infty} \left[ \frac{x^o_{[n]}}{n!} \, \hat{t}^{\,n}\right]\\
    &= \sum_{n=0}^{k-1} \left[ \frac{x^o_{[n]}}{n!} \, \hat{t}^{\,n}\right ] + \sum_{n=0}^{\infty} \left[ \frac{x^o_{[n+k]}}{(n+k)!} \, \hat{t}^{\,n}\right]  \, \hat{t}^{\,k}\\
    &= \sum_{n=0}^{k-1} \left[ \frac{x^o_{[n]}}{n!} \, \hat{t}^{\,n}\right ] + \sum_{n=0}^{\infty} \left[ \frac{n!\,x^o_{[n+k]}}{n!(n+k)!} \, \hat{t}^{\,n}\right]  \, \hat{t}^{\,k}\\
    &= \sum_{n=0}^{k-1} \left[ \frac{x^o_{[n]}}{n!} \, \hat{t}^{\,n}\right ] + \sum_{n=0}^{\infty} \left[ \frac{\tilde{x}^o_{[n]}}{n!} \, \hat{t}^{\,n}\right]  \, \hat{t}^{\,k}\\
    &= \sum_{n=0}^{k-1} \left[ \frac{x^o_{[n]}}{n!} \, \hat{t}^{\,n}\right ] + \hat{x}(\hat{t})  \, \hat{t}^{\,k}
    \end{split}
\end{equation}
where $\tilde{x}^o_{[n]} := \frac{n!\, x^o_{[n+k]}}{(n+k)!}$ and $\hat{x}(\hat{t})$ is the unknown higher-order terms state function.

After the Taylor operator, a vector of state prediction $\bar{x}_h$ can be obtained via differentiation of the inferred primitive states $x^h_{[0]}$ w.r.t. the elapse time $\hat{t}$ as described in Eq. \eqref{eq:hion-state-diff}. A subset of these states $x_h$ is outputted as the system state prediction.

The control definition serves as the third main stage. In this stage, the control $u_h$ is defined using the ODEs and the vector of high-order state estimations $\bar{x}_h$. It is described in a way that no residual exists between the estimation of the states and the dynamics of the system for ODEs that contain a control point. The control can be interpreted as a relation of the system states.

\begin{equation}
    u_h := \mathcal{F}_u(t, \bar{x}_h)
\end{equation}

Now that a state estimation is inferred and the control associated with it is defined, the final stage is to generate a set of co-states to guide the optimality. The co-state generator is a neural network, that computes the Lagrangian multiplier needed to satisfy PMP. In the experiments presented, it is also modeled by a MLP,
\begin{equation}
    \label{eq:mlp-definition-costate}
    \begin{split}
        \lambda_h &= h_{\lambda}(\hat{t}, x_o, x_r, \bar{x}_h, u_h; \mu_\lambda) = (h_k \circ \sigma \circ ... \circ h_2 \circ \sigma \circ h_1) (\hat{t}, x_o, x_r, \bar{x}_h, u_h) \\
        h_j(z) &=  W_j z + b_j \\
    \end{split}
\end{equation}
where $\mu_\lambda$ is the set of learnable parameters for the co-states generator and $\sigma(\cdot) = \text{SiLU}(\cdot)$.

\subsection{Pontryagin's Maximum/Minimum Principle and Learning Algorithm}

Pontryagin's Maximum/Minimum Principle (PMP) was relied upon to guide the parameters optimization of our neural controller for its computational efficiency. Devised by Lev Pontryagin in the Soviet Union, PMP presents a set of necessary conditions for optimal control of deterministic dynamics \citep{ma2021optimal, hwang2022problem, todorov2012lecture, todorov2006optimal}. It can be understood through the method of Lagrange multipliers where a function, known as the Hamiltonian $H$, is defined that must be maximized/minimized. The function ensured that the system dynamics are respected when considering the Lagrangian $L$ optimization. The optimal control of the system given the Lagrangian is determined by finding the control $u$ that is at an extreme of the Hamiltonian $H$. Given PMP, necessary conditions for optimally of a controller that drives a system towards a reference state $x_r(t)$ in a TPVBP are:
\begin{equation}
    \label{eq:pmp-conditions}
    \begin{array}{c}
    \dot{\lambda}(t)^T = - \frac{\partial}{\partial x} H(t, x(t), u(t), \lambda(t))\\
    \lambda(t_f)^T =  \frac{\partial}{\partial x} G(t_f, x(t_f), \nu) \\
    \frac{\partial}{\partial u} H(t, x(t), u(t), \lambda(t)) = 0 \\
    \end{array}
\end{equation}
where $\nu$ a free variable, and,
\begin{equation}
\label{eq:Hamiltonian-h}
    H(t, x(t), u(t), \lambda(t)) := L(t, x(t), x_r(t), u(t)) + \lambda(t)^T f(t, x(t), u(t))
\end{equation}
\begin{equation}
\label{eq:Hamiltonian-g}
    G(t_f, x(t_f), \nu) := \nu^T \left ( x(t_f) - x_r(t) \right ) 
\end{equation}.


In addition to these conditions for optimality, we must also ensure that dynamics of the system are maintained in estimation. Hence, we must also include the environment's ODEs and boundary as necessary conditions for dynamics accuracy:

\begin{equation}
\begin{array}{c}
x(0) = x_o\\
x(t_f) = x_r(t)\\
\mathcal{F}(t, x(t), u(t)) = 0\\
\end{array}
\end{equation}
.

Now that all the conditions for accurate and optimal trajectories are known. They provide us a set of loss functions by mean-square error (MSE) to guide our learning algorithms:
\begin{subequations}
    \label{eq:hion-loss}
    \begin{gather}
        \label{eq:-hion-loss-init}
        \frac{1}{n_x}\left\lVert x_h(0) - x_o \right\lVert^2\\
        \label{eq:-hion-loss-term}
        \frac{1}{n_{x_r}}\left\lVert x^r_h(t_f) - x_r \right\lVert^2\\
        \label{eq:-hion-loss-ode}
        \frac{1}{n_\mathcal{F}}\left\lVert \mathcal{F}(\hat{t}, x_h, u_h) \right\lVert^2 \\
        \label{eq:-hion-loss-cos}
        \frac{1}{n_x} \left\lVert \dot{\lambda_h}^T + \frac{\partial}{\partial x} H(\hat{t}, x_h, u_h, \lambda_h) \right\lVert^2 \\
        \label{eq:-hion-loss-te-cos}
        \frac{1}{n_{x_{!r}}} \left\lVert \lambda_h^{!r}(t_f)^T \right\lVert^2 \\
        \label{eq:-hion-loss-u}
        \frac{1}{n_u} \left\lVert \frac{\partial}{\partial u} H(\hat{t}, x_h, u_h, \lambda_h) \right\lVert^2
    \end{gather}
\end{subequations}
where $x^r_h$ and $\lambda^{!r}_h(t_f)$ refers to the associated states and not associated co-states with a reference state, respectively, and $n_{g}$ defines the dimensionality of a vector $g$.

However, not all of these losses are needed to train a Hion controller. In the aforementioned T-mano architecture, the Taylor operator ensures that the initial boundary condition for the states is maintained in inference, allowing \eqref{eq:-hion-loss-init} not to be required during training. If the problem at hand does not rely on a reference state, then \eqref{eq:-hion-loss-term} is also not needed. Additionally, the control definition ensures that any ODE that contains a single control point is satisfied, thus this subset of ODEs do not need to be included in \eqref{eq:-hion-loss-ode}. Lastly, the co-state generator may also be amended to use the Taylor operator, fixing the co-state terminal condition. This would remove the need for \eqref{eq:-hion-loss-te-cos}.

Following the conversion set forth by PINN models and applying a gradient decent based optimization, Algorithm \ref{alg:hion-training} is used to train and fine-tune our Hion controllers through the differential equations \eqref{eq:hion-loss}. The algorithm initially consists of sampling a number of observed states with individuality dependent references. They represent the start and end of random trajectories for which solutions need to be found at a given epoch. Three time interval are then considered for evaluating the model. First, the controller is evaluated when the elapse time $\hat{t} = 0$. The evaluation is used to measure loss \eqref{eq:-hion-loss-init} at which the inferred state $x_h(0)$ must equal the observed state $x_o$. The second time interval is the terminal time $\hat{t} = t_f$ at which losses \eqref{eq:-hion-loss-term} and \eqref{eq:-hion-loss-te-cos} are measured. It represents the time at which the reference state must be reached and the boundary condition for the co-states satisfied. Lastly, we consider transient time intervals $\hat{t}$ sampled from an uniform distribution between the boundaries. At these intervals, the transient losses \eqref{eq:-hion-loss-ode}, \eqref{eq:-hion-loss-cos}, and \eqref{eq:-hion-loss-u} are measured to ensure that the dynamics and optimality conditions are learned. Note that \eqref{eq:-hion-loss-ode}, \eqref{eq:-hion-loss-cos}, and \eqref{eq:-hion-loss-u} may optionally be applied at $\hat{t} = 0$ or $\hat{t} = t_f$ as well. The evaluations of these losses are then used to train the controller via backpropagation and a chosen optimizer.

\begin{algorithm}[ht]
\caption{Training and fine-tuning algorithm for Hion controllers}
\label{alg:hion-training}
\begin{algorithmic} 
    \Require $n_{E}$, $\mu$, $h(\cdot; \mu)$, $\mathcal{F}(\cdot)$, $f(\cdot)$, $L(\cdot)$, $t_f$, $\mathbf{X_o}$, $\mathbf{X_r}$
    \For{$n_E$ epochs}
        \State Sample $x_o$ and $x_r$ from $\mathbf{X_o}$ and $\mathbf{X_r}$, respectively
        \State Sample auxiliary elapsed time $\hat{t}$ from $\mathcal{U}(0, t_f)$
        \State Compute $x_h(0),\, u_h(0),\, \lambda_h(0) = h(0, x_o, x_r; \mu)$
        \State Compute $x_h(t_f),\, u_h(t_f),\, \lambda_h(t_f) = h(t_f, x_o, x_r; \mu)$
        \State Compute $x_h,\, u_h,\, \lambda_h = h(\hat{t}, x_o, x_f; \mu)$
        \State Evaluate and sum the MSE losses relevant to the problem in equation \eqref{eq:hion-loss} into $f_h$
        \State Compute the gradient of $f_h$ w.r.t. $\mu$.$\left(\text{i.e.,}\frac{\partial f_h}{\partial \mu}\right)$
        \State Update $\mu$ using $\frac{\partial f_h}{\partial \mu}$ and an optimizer \Comment{Adam was used for results shown}
    \EndFor
\end{algorithmic}
\end{algorithm}

\section{Experiments and Results}
\label{sect:hion-results}

To evaluate the capabilities of the proposed architecture T-mano, we will use it to find solutions to the \textbf{problem of interest} for a set of dynamical systems (i.e, environments).

\subsection{Dynamical Systems}

Two systems are considered -- one linear and one non-linear. The linear system in a basic unstable second-order linear system that serves as a baseline. The non-linear system is a Van der Pol oscillator. It models an oscillating system with non-linear damping.

\subsubsection*{Second-Order Linear System}

Second-order linear systems model a range of different phenomenons in nature. One such phenomenon is the classical spring-mass problem that explores the movement of a mass attached to a spring as a control force is applied to it over time. For our experiments, we have chosen an unstable system of a mass that is not attached to spring. It is described by the following ODE,
\begin{equation}
    \mathcal{F}(t, \bar{x}(t), u(t)) = \ddot{x}_{[0]} - u = 0
\end{equation}
and its corresponding dynamics function,
\begin{equation}
    \dot{x} = f(x, u) = \begin{bmatrix} x_{[1]} \\ u\end{bmatrix}
\end{equation}
where $x_{[0]}$ is the primitive state of the second-order linear system (i.e., displacement of the mass) and $u$ is an applied force over time.

For the observed state distribution used to train the Hion controller, we have chosen the following random state vector:
\begin{equation}
    \mathbf{X_o} = \begin{bmatrix}
     X^o_{[0]} \\ X^o_{[1]} \end{bmatrix} \sim \begin{bmatrix}
     \mathcal{U}(-5, 5) \\ \mathcal{N}\left (0, 1 \right) \end{bmatrix}
\end{equation}
where $X^o_{[0]}$ and $X^o_{[1]}$ are the random variables from which the observed displacement $x^o_{[0]}$ and velocity $x^o_{[1]}$ states are sampled from. For the reference state's random variable, we have opted to add a white Gaussian noise with a standard deviation of $1$. A terminal time $t_f$ of 2 seconds was considered for the system.

To guide the transient behavior of the system, we have opted for a quadratic Lagrangian function that minimizes that control applied and velocity state over time:
\begin{equation}
    \label{eq:linear-cost}
    L(t, x(t), x_r(t), u(t)) = \frac{1}{2} u(t)^2 + x_{[1]}(t)^2
\end{equation}.

It should be noted that the system is invariant under all transformations on its sole primitive state $x_{[0]}$ given Definition \ref{de:dynamic-invariant}. As such, for the experiments presented, the invariant mask was applied to fix the primitive observed state $x^o_{[0]}$ to zero before going to the state generator.

\subsubsection*{Van der Pol Oscillator}

 Van der Pol oscillators were introduced by Balthasar van der Pol to describe the change in current inside a triode that is part of an electronic circuit \cite{hafeez2015analytical, abell2023chapter}. It is second-order, non-conservative, oscillating system with non-linear damping properties.  However, Van der Pol oscillators are now used to describes a wide range of oscillatory processes in a variety disciplines including, but not limited to, modeling irregular heart rate for pace maker design. It is standard benchmark for optimal control by neural controller \cite{antonelo2024physics, andersson2012dynamic}. For our experiments, we have chosen the following variant with the ODE,
\begin{equation}
    \mathcal{F}(t, \bar{x}(t), u(t)) = \ddot{x}_{[0]} - \left (1 - x_{[0]}^2 \right) \dot{x}_{[0]} + x_{[0]} - u = 0
\end{equation}
with its corresponding dynamics function $f$,
\begin{equation}
    \dot{x} = f(x, u) = \begin{bmatrix} x_{[1]} \\ \left (1 - x_{[0]}^2 \right) x_{[1]} - x_{[0]} + u \end{bmatrix}
\end{equation}
where the primitive state $x_{[0]}$ of the system represents current and $u$ is an external force.

For the observed state distribution used to train the Hion controller, we have also decided on the following random state vector:

\begin{equation}
    \mathbf{X_o} = \begin{bmatrix}
     X^o_{[0]} \\ X^o_{[1]} \end{bmatrix} \sim \begin{bmatrix}
     \mathcal{U}(-5, 5) \\ \mathcal{N}\left (0, 1 \right) \end{bmatrix}
\end{equation}
where $X^o_{[0]}$ and $X^o_{[1]}$ are the random variables from which the observed current $x^o_{[0]}$ and change in current w.r.t. time $x^o_{[1]}$ states are sampled from, respectively. For the reference state's random variable, we have also opted to add a white Gaussian noise with a standard deviation of $1$. A terminal time $t_f$ of 5 seconds was considered.

For the transient behavior of the Van der Pol oscillator, we considered two distinct set of Lagrangian functions. When trying to reduced the velocity of the system, we used the following quadratic equation:
\begin{equation}
    \label{eq:van-cost-1}
    L(t, x(t), x_r(t), u(t)) = \kappa\, x_{[1]}(t)^2
\end{equation}
and, when trying to improve tracking of a reference signal, we used
\begin{equation}
    \label{eq:van-cost-2}
    L(t, x(t), x_r(t), u(t)) = \kappa\, (x_{[0]}(t) - x_r(t))^2
\end{equation}
where $\kappa$, in both sets, is a scalar hyper-parameter that regulates the intensity of the transient cost.

\subsection{Two Point Boundary Value Problem} \index{Two point boundary value problem (TPBVP)}

A two point boundary value problem (TPBVP) is an optimal control problem where the system is required to satisfy boundary conditions at both an initial and final time interval. Hion controllers, in their most basic form, can solve these TPBVP where the terminal time $t_f$ is the time by which the final state must be reached. We can, additionally, not only drive the system towards our desired final state, but also specify the transient behavior we would like the system to have.

Figure \ref{gr:hion-tpbvp} illustrates the solution for the aforementioned systems obtained by Hion controllers using the T-mano architecture. Figure \ref{gr-hion-tpbvp-1} shows the solution for a second-order linear system when tasked with minimizing control effort and speed. Figure \ref{gr-hion-tpbvp-2} presents the solution for a Van der Pol oscillator when tasked solely with minimizing speed. In both cases, the system was driven to the final state (i.e., the reference state or $1$) while maintaining the specified transient profile. For instance, in the case of the Van der Pol oscillator, the system was driven to a constant speed, which was then maintained until the system reached the final state. This represents an optimal solution to maintain minimal speed.

\begin{figure}[ht]
  \centering
    \subcaptionbox
      {Second-order linear system.\label{gr-hion-tpbvp-1}}
      {\bfseries \includegraphics[width=0.49\textwidth]{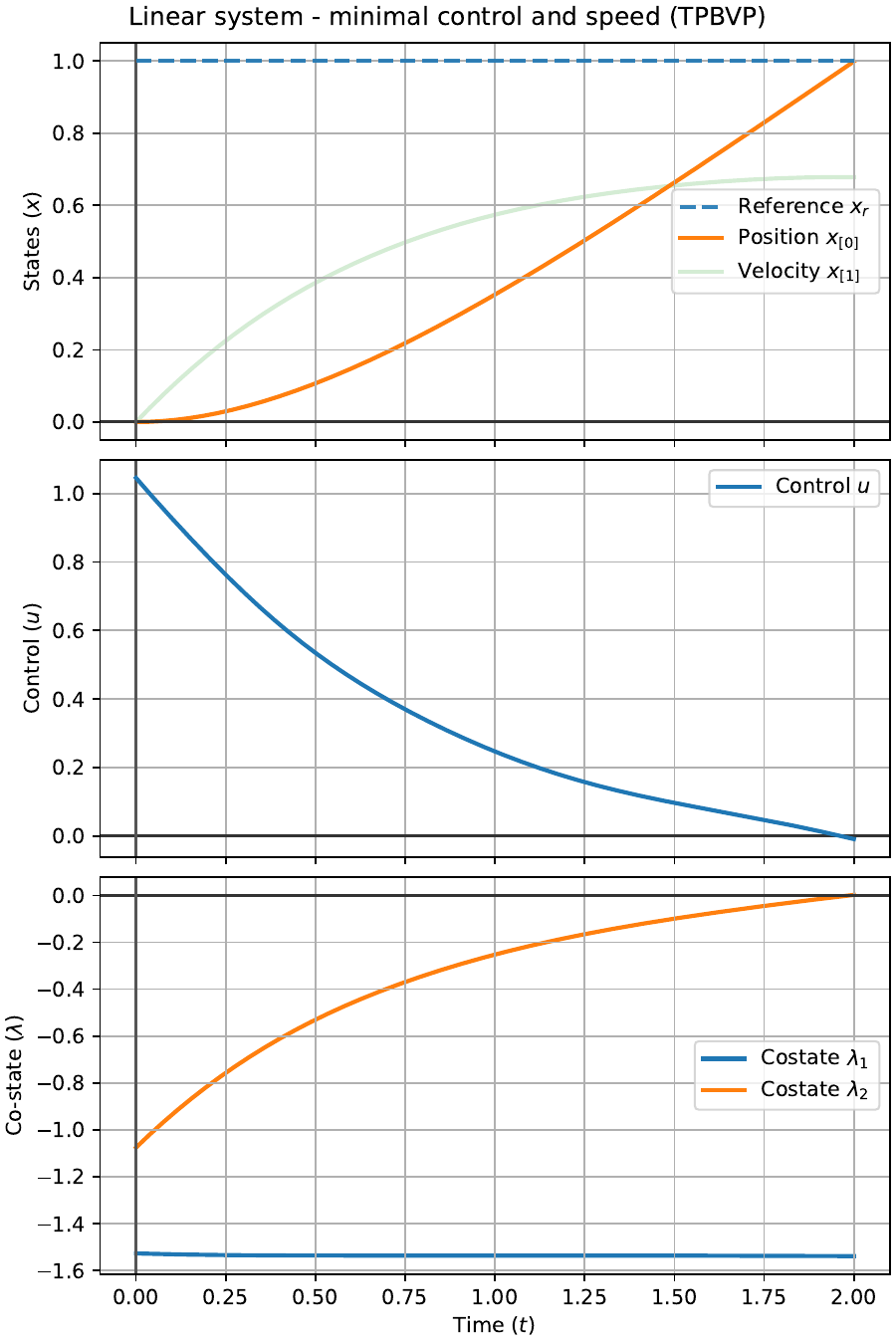}}%
    \subcaptionbox
      {Van der Pol oscillator.\label{gr-hion-tpbvp-2}}
      {\bfseries \includegraphics[width=0.49\textwidth]{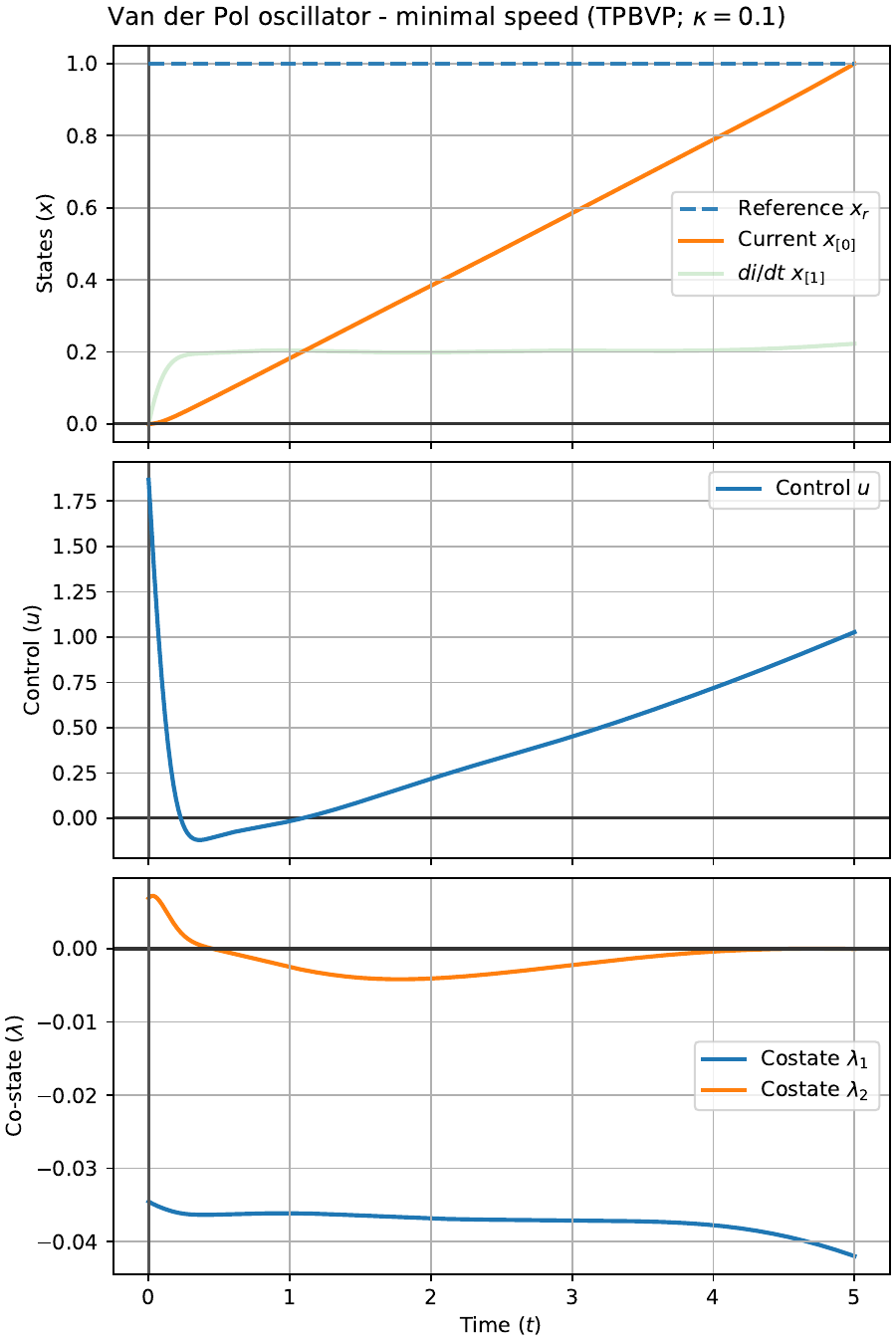}}%
  \caption{Two point boundary value problem solutions using the T-mano architecture.}
  \label{gr:hion-tpbvp}
\end{figure}

\subsection{Closed-Loop Control}

Hion controllers can be extended to solve control strategies for closed-loop systems. Given the model's ability to solve TPBVP for variable observed states and reference states, and estimate future states, we can consider closed-loop systems where there is a gap between observation of the state in the environment. This is similar to systems encountered in MPC problems and real-world scenarios where sensors only allow for sparse sampling of the system states. Unless otherwise stated, note that from now on each vertical tick on a plot indicates when a new state is observed (i.e., sampled from the simulation environment).

Figure \ref{gr:hion-result-u-velocity} considers the second-order linear system when controlled to have minimal control and speed as described by the cost function in \eqref{eq:linear-cost}. As seen, T-mano is capable of driving the system to a reference state while maintaining the optimality requested between sampling periods (or phases). Additionally, as the model iteratively obtains the solution to a TPBVP, at the end of each iteration, the reference state is reached unless the reference changes before the next state can be observed. If the reference state $x_r$ is updated before a new state can be observed from the actual environment, then the current estimated state $x_h$ can be assumed to be a new observed state $x_o := x_h$ and the elapsed time $\hat{t}$ is reset to $0$. This behavior can be seen when the reference was updated to $0$ at $t=15$ before the next sample of the environment state was captured in the next tick. In this instance, the model moved toward the new reference, and when a new observed state was captured from the actual system, it readjusted its control once again. It should be noted that $\Delta$ denotes the sampling period.

\begin{figure}[ht]
  \centering
  \includegraphics[width=\textwidth]{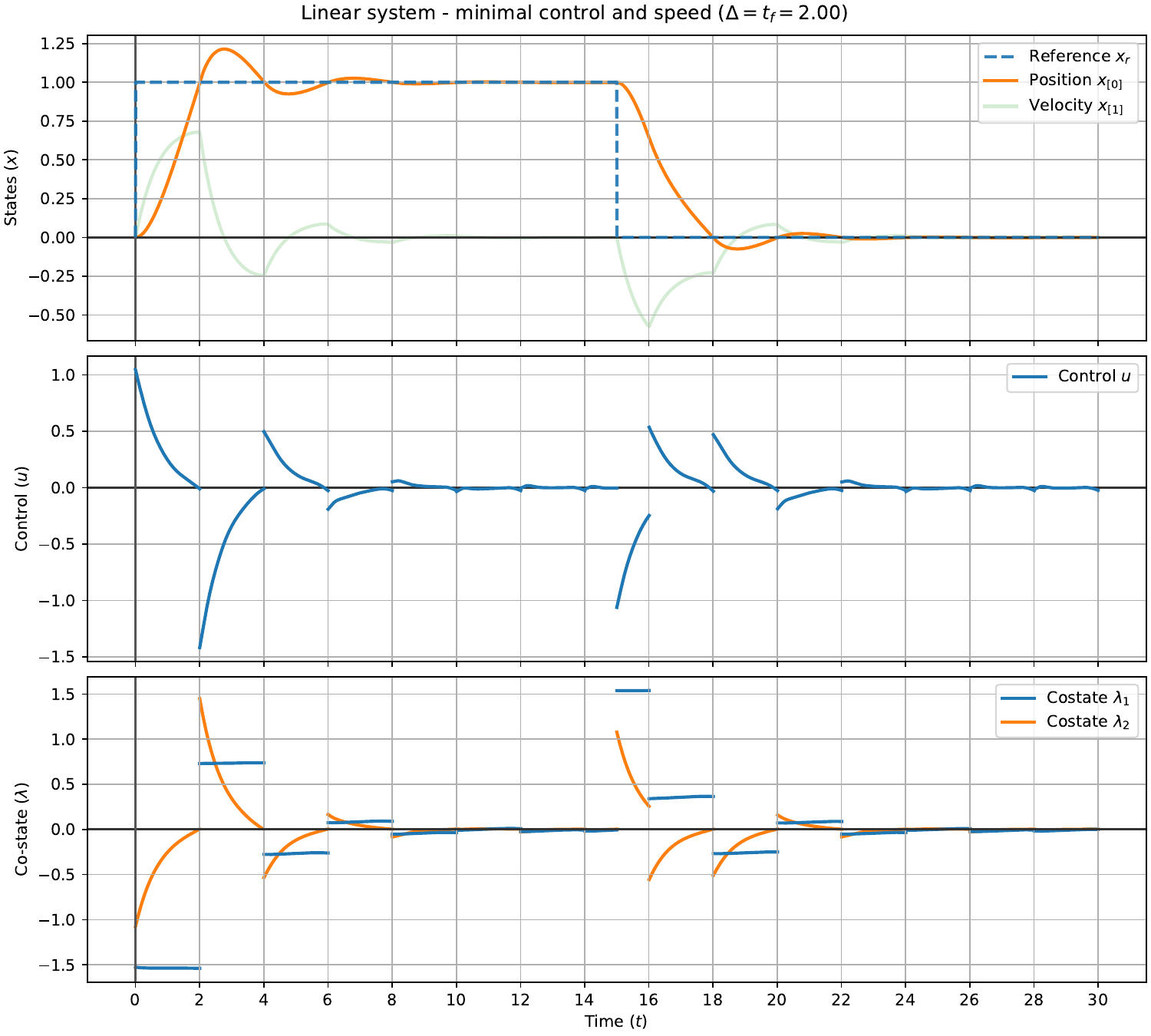}
  \caption{Linear system controlled to minimize the necessary control applied and velocity as a priority.}
  \label{gr:hion-result-u-velocity}
\end{figure}

Figure \ref{gr:hion-result-velocity} presents what occurs when controlling a non-linear Van der Pol oscillator to have minimal speed between observed states. The desired behavior is represented in \eqref{eq:van-cost-1}. As seen, T-mano controls the system to have just the constant speed needed between observed states to reach the reference state before a new state is sampled. As aforementioned, this is a known optimal solution to the problem. In the plot, it can also be seen that a maneuver with the controls occurs whenever a new state is sampled. It is speculated that this occurs because, while the reference state may be obtained by the terminal time $t_f$, some momentum exist, necessitating a non-linear maneuver to cancel it and stay on the desired course. To a lesser extent, some error between the numerical simulation of the environment and the state estimation emerges, and thus when the numerical simulation is sampled, the model corrects itself to eliminate the error.

\begin{figure}[ht]
  \centering
  \includegraphics[width=\textwidth]{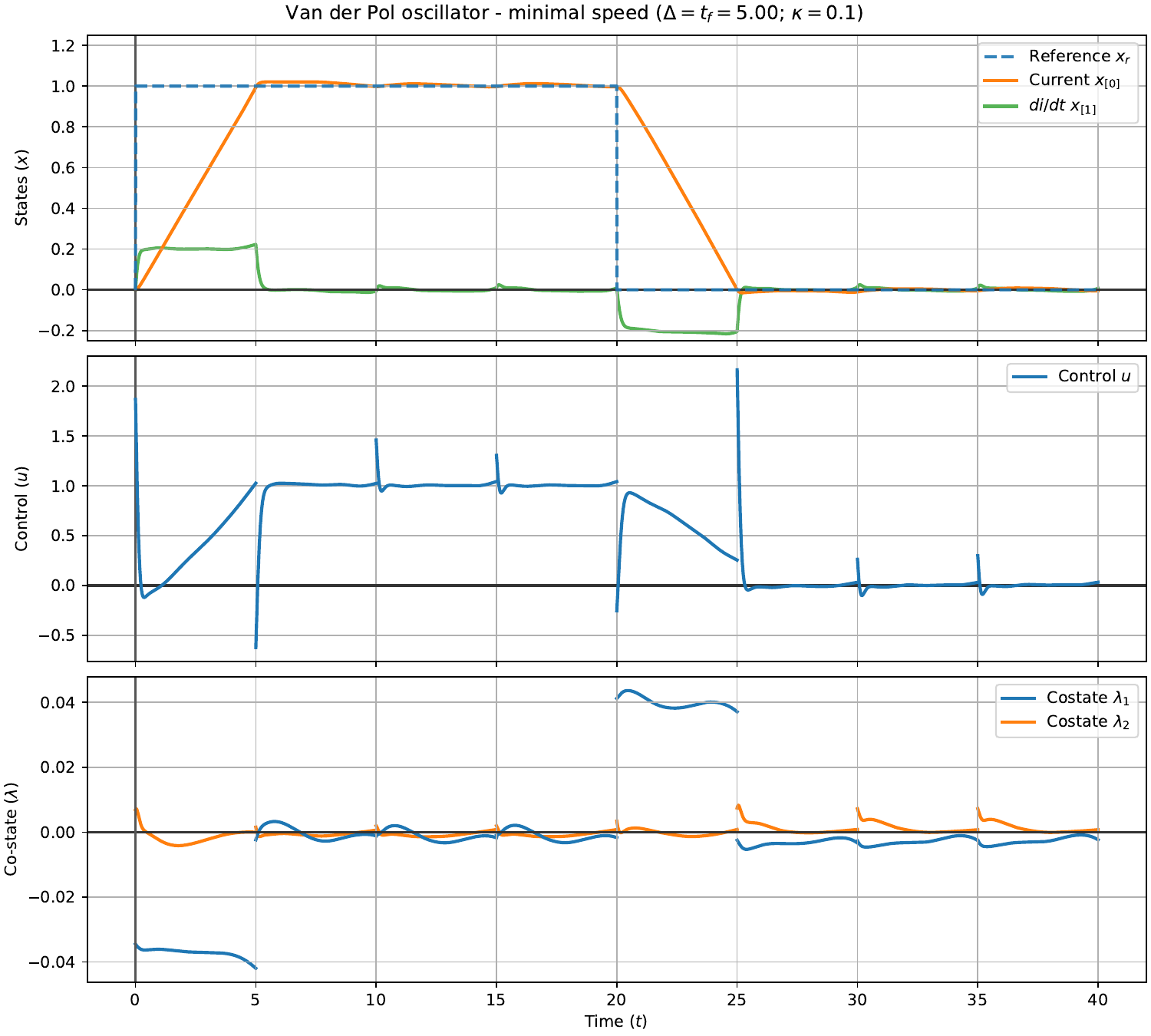}
  \caption{Van der Pol oscillator controlled to have minimal speed between sampling periods as a priority.}
  \label{gr:hion-result-velocity}
\end{figure}

In Figure \ref{gr:hion-result-track}, the tracking problem is explored, where the model is tasked with prioritizing the tracking of the reference state with no emphasis on the transient velocity. The objective to prioritize tracking of the reference is embedded using the Lagrangian equation \eqref{eq:van-cost-2}. As illustrated, the T-mano model quickly drives the current state to the reference before the new state is observed and maintains it. This behavior is a departure from the one seen in Figure \ref{gr:hion-result-velocity}. However, the maneuvering observed when a new state is sampled can still be seen due to the same conditions discussed previously. This time, however, the magnitude of the control is larger, as it is desired that the model bring the system to the reference state rapidly.

\begin{figure}[ht]
  \centering
  \includegraphics[width=\textwidth]{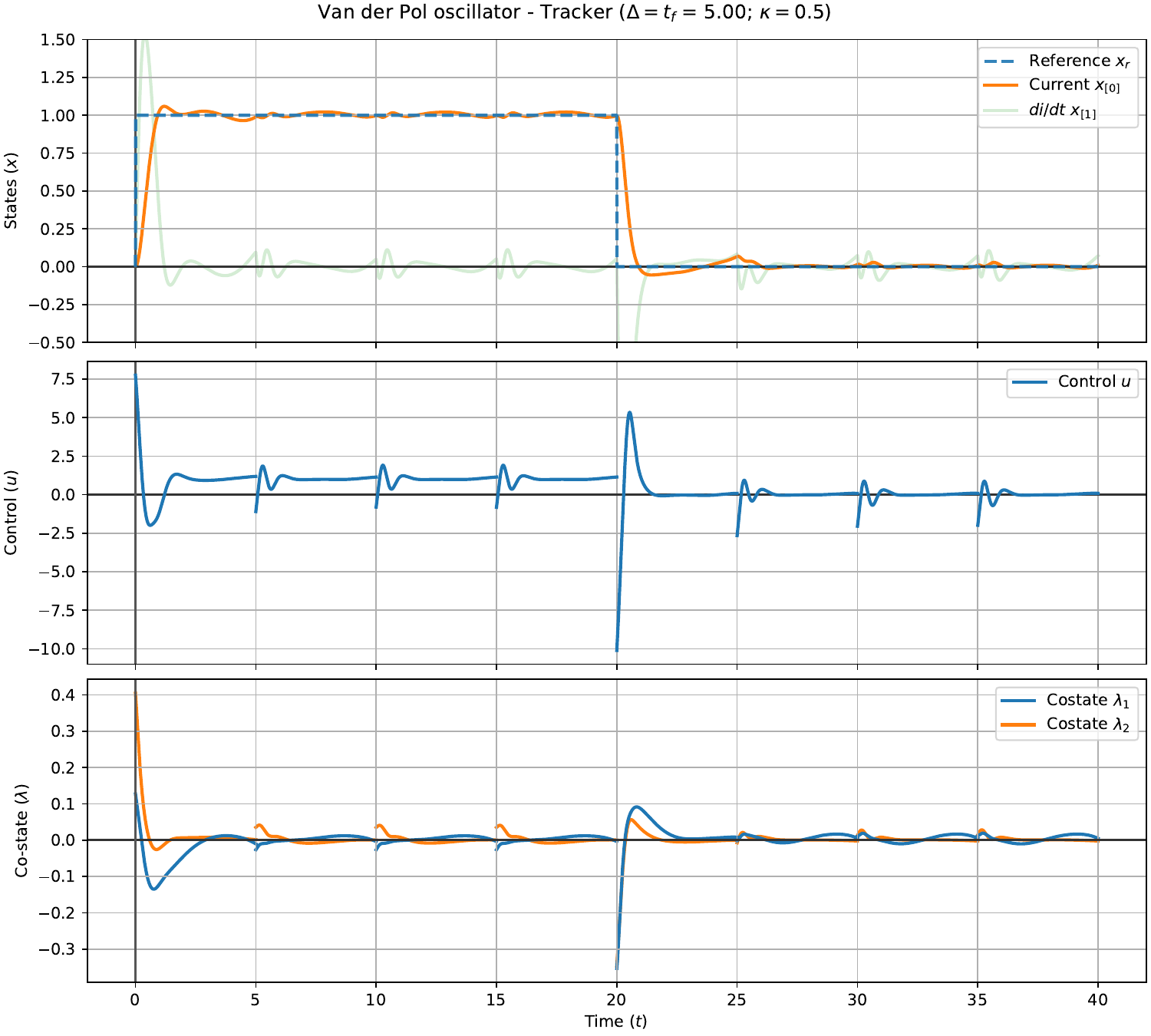}
  \caption{Van der Pol oscillator controlled to follow the reference as a priority.}
  \label{gr:hion-result-track}
\end{figure}

\subsubsection{Receding Horizon Control}
\index{Receding horizon control}

Hion controllers can also be used to drive a system towards a reference state using a model-predictive control (MPC) or receding horizon control scheme, where the terminal time $t_f$ defines the size of the prediction and control horizon. When the sampling period is shorter than the trained terminal time $t_f$, the control is based on the concept of steering the system towards a goal at the end of a receding horizon. Note that under sampling may impact the optimality of the controller. Additionally, while we showcase states sampled at regular intervals, a mix of sampling periods can be implemented depending on the availability of the environment state.

Figure \ref{gr:hion-sampling-linear} presents the behavior of our T-mano model under the MPC scheme for the second-order linear system. As observed, the model is capable of driving the system towards the reference state at distinct sampling periods. When the sampling period decreases, additional discontinuities in the control and co-states emerge. Interestingly, when the environment is sampled in real-time, the control and co-states exhibit smoothness. It is also observed that the controller takes longer to drive the system to the reference state as the sampling period decreases.

\begin{figure}[ht]
  \centering
  \includegraphics[width=\textwidth]{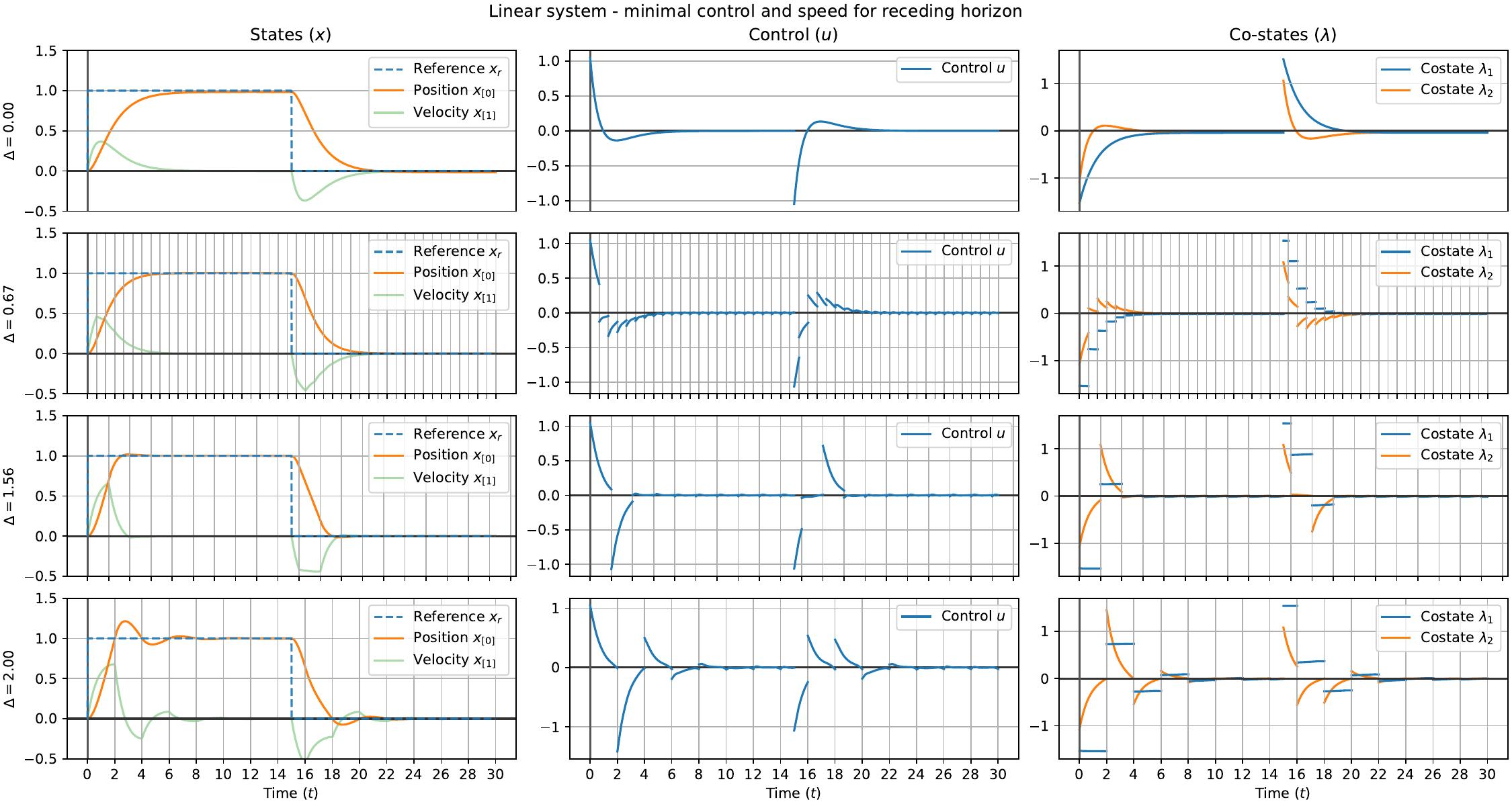}
  \caption{Same linear system controller for different sampling rates.}
  \label{gr:hion-sampling-linear}
\end{figure}

Figure \ref{gr:hion-sampling-van} showcases the results obtained for the Van der Pol tracker model when considering the MPC scheme. As observed, the controller is capable of driving the system towards the reference state of the non-linear system in most instances when under-sampled. However, unlike the linear system, some maneuvering is observed whenever the environment is sampled due to the non-linear dynamics, the desired transient characteristics, and the micro-adjustments needed due to the small accumulated error between the numerical solver and the model state estimation.  Additionally, for a sampling period of $0.0$ (i.e., without any delay in state feedback), it is observed that the model does not control the system towards the reference state. The origin of this behavior is unclear, but it has been linked to the convergence of the losses associated with optimality during training. Finally, the discontinuity and rapid movement that arise in the control when quickly sampled could be challenging to implement in real-world systems and may require smoothness.

\begin{figure}[ht]
  \centering
  \includegraphics[width=\textwidth]{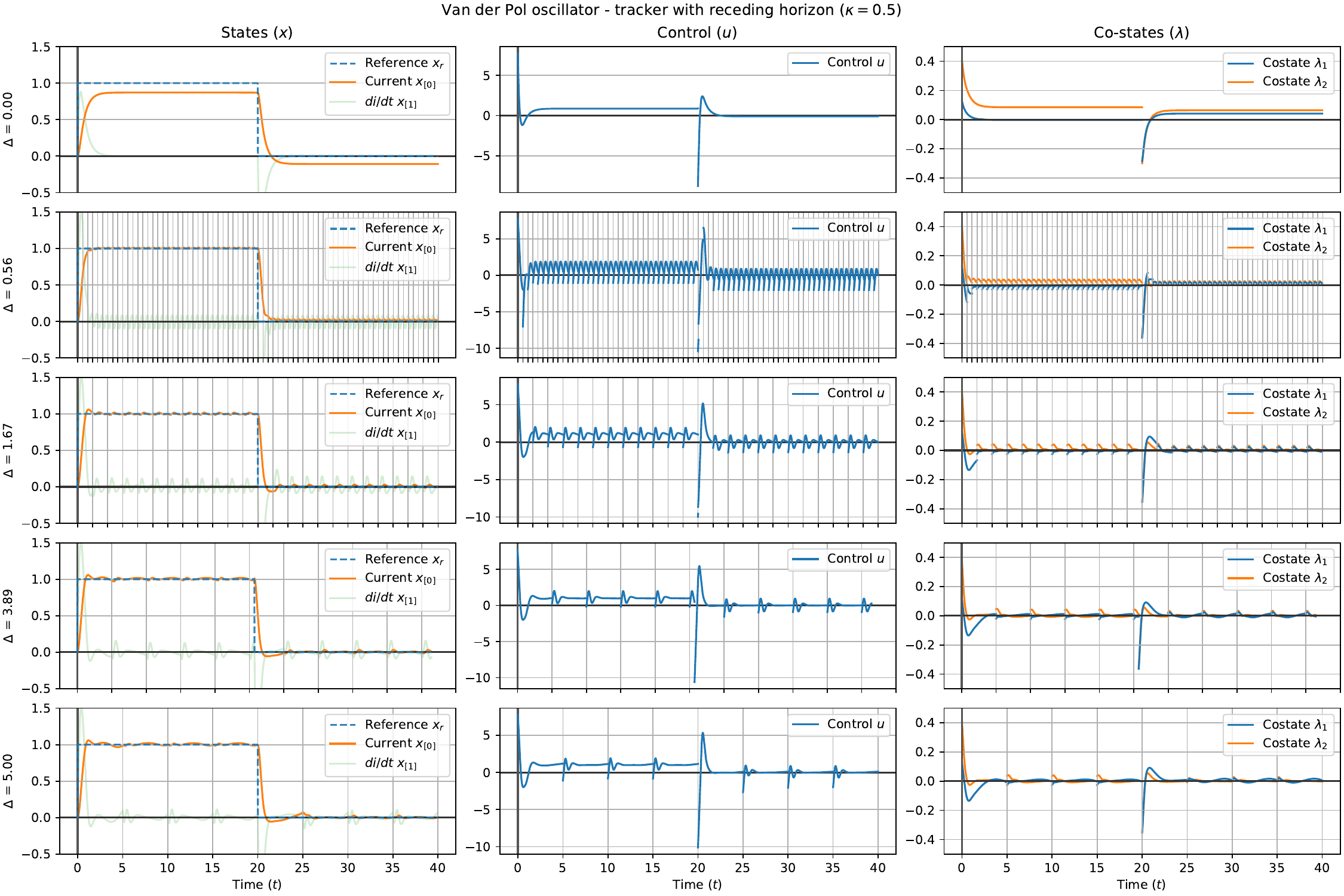}
  \caption{Same Van der Pol oscillator tracker for different sampling rates.}
  \label{gr:hion-sampling-van}
\end{figure}

\subsection{Fine-Tuning}

An attribute of Hion controllers and, in particular, our T-mano models is that once trained, the parameters $\mu$ can be easily be fine-tuned through Algorithm \ref{alg:hion-training}. It can be fine-tuned for new transient behaviors (given a change to the Lagrangian function), constants in the dynamics, terminal time $t_f$, and state distributions.

Figure \ref{gr:hion-lagrangian-van} demonstrate variant T-mano models and their change in behavior when fine-tuned to distinct but similar Lagrangian functions. For these results, the form of the Lagrangian functions tested remained the same but its optimality intensity was modulated by the hyper-parameter $k$ in Eq. \eqref{eq:van-cost-1}. As previously discussed, this Lagrangian function seeks to reduce the velocity of the oscillator between sampling periods. When k is decreased, it can be observed weaker optimality for maintaining reduced speed between sampling, specially once the reference state is reached. It can also be observed that the co-states -- which can heuristically be understood as the rate of change towards optimality w.r.t. time -- is lower the lesser the intensity for reduced speed is. Lower values suggest that the model has less interest in maintaining its optimality. This result shows the T-mano can be fine-tuned for new transilient behavior.

\begin{figure}[ht]
  \centering
  \includegraphics[width=\textwidth]{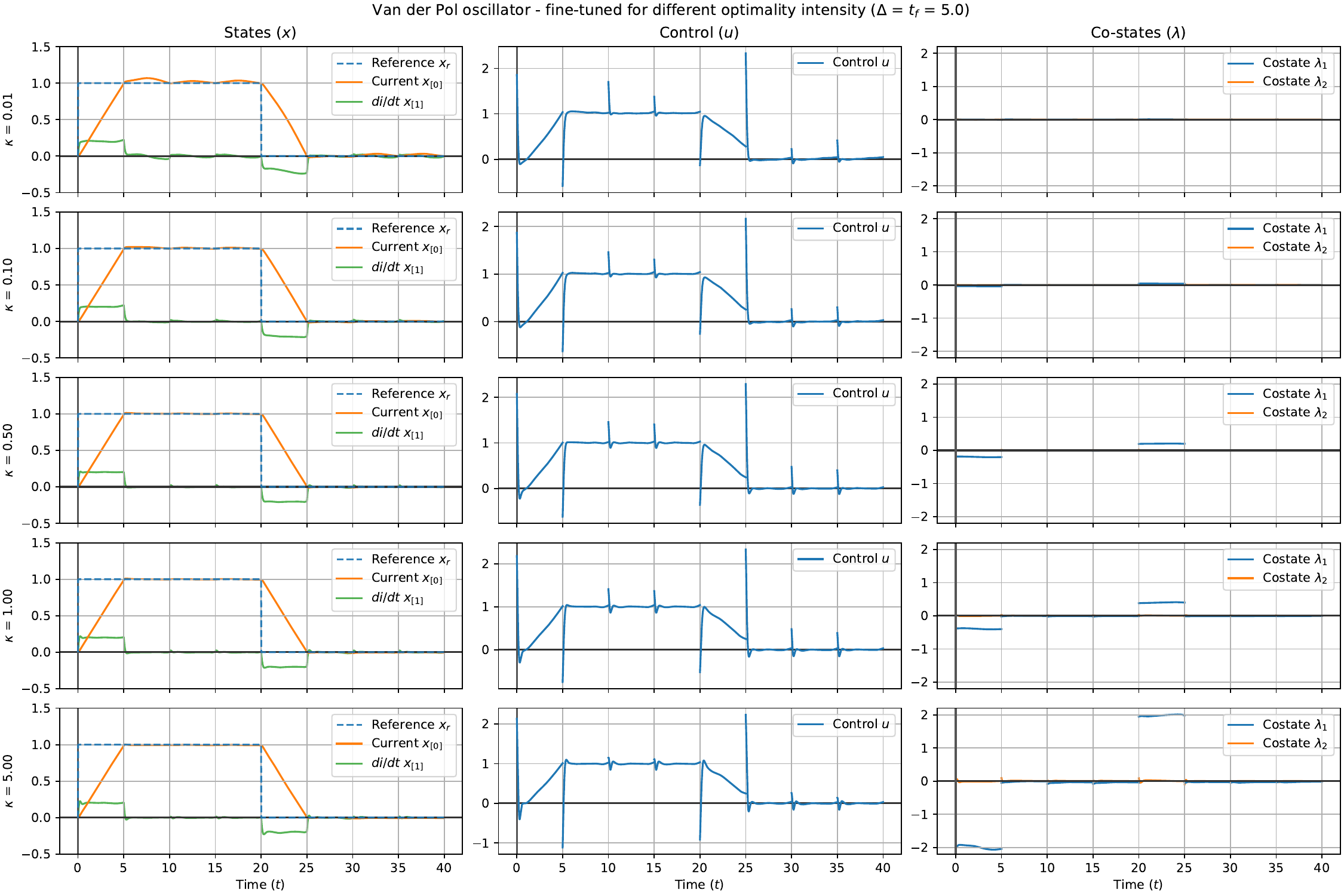}
  \caption{Van der Pol oscillator fine-tuned for different intensity in transient response.}
  \label{gr:hion-lagrangian-van}
\end{figure}

It should also be highlighted that the tracking controllers presented in Figure \ref{gr:hion-result-track} and \ref{gr:hion-sampling-van} were fined-tuned from the model trained to have minimal velocity between sampling in Figure \ref{gr:hion-result-velocity}.

\subsection{Comparison}

To evaluate the efficacy of our model, a comparative analysis was conducted against established model-predictive controllers: Successive Linear Model Predictive Control (SLMPC) \citep{zhakatayev2017successive} and Physics-Informed Neural Control (PINC) \citep{antonelo2024physics}. These controllers were tasked with following a reference signal while maintaining minimal speed within a Van der Pol oscillator. For each controller, the system states and applied controls were recorded for comparison. The results are presented in Figure \ref{gr:mpc-compare}.

\begin{figure}[ht]
  \centering
  \includegraphics[width=\textwidth]{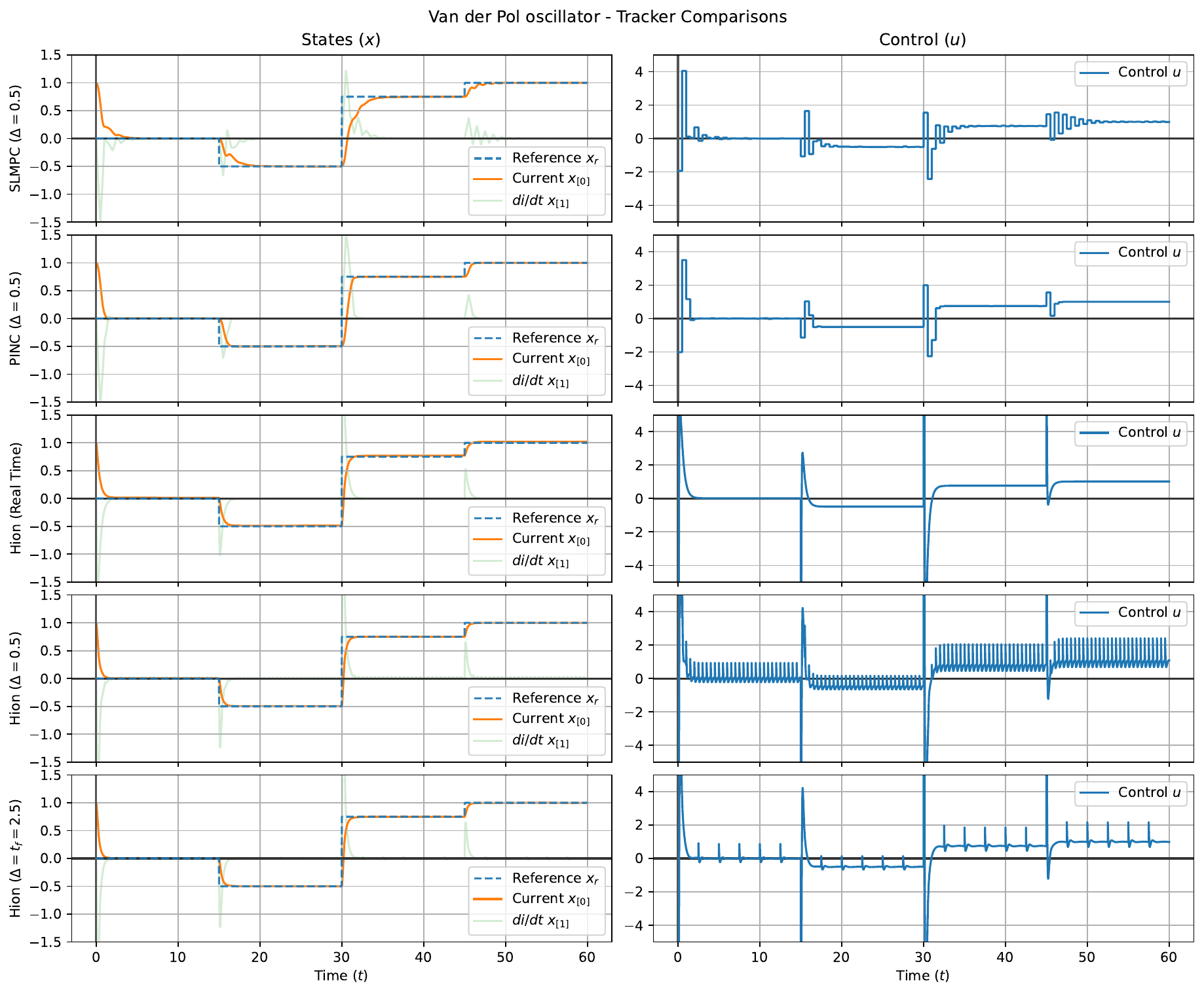}
  \caption{Comparison of model-predictive controllers in tracking a reference signal within a Van der Pol oscillator.}
  \label{gr:mpc-compare}
\end{figure}

As illustrated in Figure \ref{gr:mpc-compare}, SLMPC exhibited the poorest performance in tracking the reference signal. This controller demonstrated a delayed response to changes in the reference and generated significant oscillations during these transitions. These performance issues can be attributed to its reliance on inaccurate future state estimations (derived from system linearization) and zero-order hold control strategies. On the other hand, PINC exhibited enhanced performance, evidenced by the reduction in oscillations, quicker response, and smoother control. This improvement is primarily driven by the neural network-based surrogate model for non-linear state estimation. However, the dependence on zero-order hold strategies restricts PINC from achieving the most optimal control. Our Hion controller showed the best performance, although frequent discontinuities and larger control spikes were observed.

The optimality and tracking capability of these models were investigated by measuring the cumulative cost and reference error of the simulated trajectories. The cost function employed to define the transient goal of each controller is represented by the Lagrangian $L$: \begin{equation} \label{eq:compare-cost} L(t, x(t), x_r(t), u(t)) = (x_{[0]}(t) - x_r(t))^2 + \frac{x_{[1]}(t)^2}{10} \end{equation} which prioritizes tracking the reference signal while minimizing current variation. Table \ref{tb:compare} presents the results obtained for the runs depicted in Figure \ref{gr:mpc-compare}. Notably, all controllers, including SLMPC and PINC, utilized a horizon window of $2.5$ time units.

\begin{table}[htb]
    \caption{Performance metrics for various model-predictive controllers. Hion (Our) demonstrates superior optimality (lower cost $J$) and tracking capability compared to SLMPC \citep{zhakatayev2017successive} and PINC \citep{antonelo2024physics}.}
    \label{tb:compare}
    \centering
    \begin{tabularx}{\textwidth}{@{}|l|c|M|M|M|@{}}
        \hline
        \textbf{MPC}  & \textbf{Sampling period $\Delta$} & \textbf{Cost $J$} & \textbf{$ \int \| x_{[0]}(t) - x_r(t)\|^2\, dt$} \\
        \hline
        \hline
        SLMPC & 0.5 & 2.0346 & 1.8530 \\
        \hline
        PINC & 0.5 & 1.7086 & 1.4472  \\
        \hline
        Hion (Our) & real time  & 1.0507 & 0.6914 \\
        \hline
        Hion (Our) & 0.5 & 1.0331 & 0.5832 \\
        \hline
        Hion (Our) & 2.5 ($t_f$) & \textbf{1.0296} & \textbf{0.5807} \\
        \hline
    \end{tabularx}
\end{table}

As corroborated by Figure \ref{gr:mpc-compare}, the results presented in Table \ref{tb:compare} indicate that our Hion controller, across varying sampling periods $\Delta$, outperforms SLMPC and PINC in terms of greater optimality (lower cost) and enhanced tracking performance. Our Hion controller, with a sampling period of $t_f$, exhibited the best performance, followed closely by the other approaches. As previously discussed, sampling the environment in real-time decreases a Hion controller's optimality and convergence towards a reference signal. Hence, it is understood that performance minimally degrades. Conversely, frequent sampling over short intervals results in control discontinuities and large spikes, causing recurrent oscillations in a non-linear Van der Pol oscillator. Optimal performance was observed when the system was sampled at a rate that allowed completion of a control strategy within the expected control window. This approach led to fewer discontinuities and resulted in smoother responses. The main drawback of infrequent sampling is the reliance on internal state estimation to control the system for potentially significant durations.

\section{Conclusion}

This work formalizes Hion controllers, a novel class of neural network models designed to achieve two key functionalities: estimating future states of dynamical systems and calculating the optimal control needed to reach desired states. Alongside it, we present the T-mano architecture, which ensures accurate initial conditions and system dynamics in the state estimation process, independent of the model parameters values. These models are demonstrated to have greater capability at tracking reference signal while maintaining optimal performance when compared against established models. The source code for this project is publicly available at \href{https://github.com/wzjoriv/Hion}{https://github.com/wzjoriv/Hion}.

Hion controllers (as modeled by our T-mano architecture) demonstrate significant potential for real-world applications. Firstly, they can be trained to manipulate the transient characteristics of the controlled system. For instance, industrial robots in a factory setting can be programmed with our controllers to operate at lower accelerations or speeds, enhancing safety and mitigating potential collisions. Secondly, Hion controllers estimate the future states of a system. This predictive capability empowers proactive collision avoidance and obstacle clearance strategies. It may also facilitate the coordination of multi-agent systems by anticipating the future trajectories of individual agents and their potential interaction points. Additionally, it fosters resilience against state feedback delays and disturbances by leveraging state estimation to remove noise or relying on predicted states until a new measurement becomes available. Lastly, Hion controllers can estimate higher-order time derivatives of states and controls, enabling analysis of properties such as jerk, snap, crackle, and pop, and allowing for control over their transient behavior.

However, limitations exist.  Currently, control actions may exhibit discontinuities upon observing new states due to the creation of a new control strategy without prior information. This may be problematic for real-world systems that cannot tolerate sudden changes in control inputs. A potential solution lies in incorporating information about higher-order derivatives of the previous iteration's state estimation. Another limitation pertains to control scenarios with minimal delay between state measurements. In this situation, T-mano models may not drive the system to the reference state contingent to the convergence of the model's parameters as seen in the results section.

Future research directions are promising. We aim to investigate the application of Hion controllers for control of more complex non-linear chaotic dynamical systems. The ability to predict future states also provides a foundation for collision avoidance strategies in multi-agent systems or complex environments. Additionally, T-mano models can be made small with few parameters. The exploration of their implementation in resource-constrained embedded systems is another compelling avenue. Given that higher-order state and control estimations are obtained via differentiability of the model, no additional parameters are needed to model them. Other future works will explore expanding the T-mano architecture for enhanced state estimation under noisy observed state sampling.

\vskip 0.2in
\bibliography{root}

\end{document}